\documentclass[11pt,a4paper]{amsart}
\usepackage[left=25mm,top=30mm,bottom=25mm,right=25mm]{geometry}
\usepackage{amssymb,dsfont,tensor}
\usepackage[all]{xy}
\usepackage{bm} 
\usepackage[utf8]{inputenc}
\usepackage[english]{babel}
\usepackage{graphicx}
\usepackage[pdfdisplaydoctitle=true,
      colorlinks=true,
      urlcolor=blue,
      citecolor=blue,
      linkcolor=blue,
      pdfstartview=FitH,
      pdfpagemode=UseNone,
      bookmarksnumbered=true,
      unicode=true]{hyperref}                                       
\usepackage{hyperxmp}
\usepackage{cmap}                                                   

\newtheorem{thm}{Theorem}[section]

\newtheorem{lem}[thm]{Lemma}

\theoremstyle{definition}

\newtheorem{exam}[thm]{Example}

\theoremstyle{remark}
\newtheorem{rem}[thm]{Remark}

\numberwithin{equation}{section}

\newcommand{\RR}{\mathbb{R}}                

\newcommand{\Sym}{\mathbb{S}}               
\newcommand{\espace}{\mathcal{E}}           
\newcommand{\body}{\mathcal{B}}             

\newcommand{\Emb}{\mathrm{Emb}^{\infty}}    

\newcommand{\Cinf}{\mathrm{C}^{\infty}}

\newcommand{\norm}[1]{\left\Vert#1\right\Vert}

\newcommand{\set}[1]{\left\{#1\right\}}

\newcommand{\vol}{\mathrm{vol}}             
\newcommand{\Ric}{\mathbf{Ric}}             
\newcommand{\Ein}{\mathbf{Ein}}             

\newcommand{\rd}{\mathrm{d}}
\newcommand{\rdx}{T}
\newcommand{\pp}{\mathrm{p}}

\newcommand{\bA}{\mathbf{A}}
\newcommand{\bB}{\mathbf{B}}
\newcommand{\bC}{\mathbf{C}}

\newcommand{\bF}{\mathbf{F}}

\newcommand{\bI}{\mathbf{I}}

\newcommand{\Jmat}{\bm{J}^{\text{mat}}}
\newcommand{\bK}{\mathbf{K}}

\newcommand{\MM}{\bm{M}}

\newcommand{\NN}{\bm{N}}

\newcommand{\bP}{\mathbf{P}}

\newcommand{\UU}{\bm{U}}
\newcommand{\bV}{\mathbf{V}}
\newcommand{\VV}{\bm{V}}

\newcommand{\bX}{\mathbf{X}}

\newcommand{\ba}{\mathbf{a}}

\newcommand{\ff}{\bm{f}}
\newcommand{\bg}{\mathbf{g}}

\newcommand{\nn}{{\bm{n}}}
\newcommand{\bq}{\mathbf{q}}

\newcommand{\uu}{\bm{u}}
\newcommand{\bp}{\bm{p}}

\newcommand{\xx}{\mathbf{x}}

\newcommand{\mH}{\mathcal{H}}
\newcommand{\mL}{\mathcal{L}}
\newcommand{\mM}{\mathcal{M}}
\newcommand{\mW}{\mathcal{W}}

\newcommand{\bgamma}{\bm{\gamma}}

\newcommand{\btau}{\bm{\tau}}

\newcommand{\mg}{{\mathfrak g}}

\newcommand{\id}{\mathrm{id}}
\newcommand{\Id}{\mathbf{I}}
\newcommand{\Idt}{\mathbf{I}_{3}}
\newcommand{\Idq}{\mathbf{I}_{4}}
\newcommand{\bEsh}{\hat{\pmb{\mathbb{B}}}}
\newcommand{\Eshcomp}{\mathbb{B}}
\newcommand{\bMan}{\hat{\pmb{\mathbb{M}}}}

\newcommand{\bPK}{\hat{\pmb{\mathbb{P}}}}

\newcommand{\bTN}[1]{\hat{\pmb{\mathbb{T}}}{}_{#1}}
\newcommand{\TNcomp}[1]{{\mathbb{T}_{#1}}{}}
\newcommand{\bTH}{\pmb{\mathbb{T}}{}^{\text{H}}}

\newcommand{\dd}[2]{\frac{\rd{#1}}{\rd{#2}}}
\newcommand{\pd}[2]{\frac{\partial{#1}}{\partial{#2}}}

\DeclareMathOperator{\Lie}{L} %
\DeclareMathOperator{\tr}{tr} %
\DeclareMathOperator{\grad}{grad} %
\DeclareMathOperator{\gradg}{grad^{\textit{g}}} %
\DeclareMathOperator{\dive}{div} %
\hypersetup{
    pdftitle={Classical and relativistic balance of configurational forces},
    pdfauthor={Rodrigue Desmorat, Anthony Gravouil and Boris Kolev},
    pdfsubject={MSC 2020: 74B20, 70G45, 83C10, 83C25},
    pdfkeywords={Configurational forces, Constitutive equations, Relativistic Hyperelasticity, Lagrangian formulation of General Relativity},
    pdflang=en
}
\begin{document}

\title{Classical and relativistic balance of configurational forces}

\author{R. Desmorat}
\address[Rodrigue Desmorat]{Université Paris-Saclay, CentraleSupélec, ENS Paris-Saclay, CNRS, LMPS - Laboratoire de Mécanique Paris-Saclay, 91190, Gif-sur-Yvette, France}
\email{rodrigue.desmorat@ens-paris-saclay.fr}

\author{A. Gravouil}
\address[Anthony Gravouil]{INSA Lyon, CNRS, LaMCoS - Laboratoire Mécanique des Contacts et des Structures, UMR5259, 69621 Villeurbanne, France}
\email{Anthony.Gravouil@insa-lyon.fr}

\author{B. Kolev}
\address[Boris Kolev]{Université Paris-Saclay, CentraleSupélec, ENS Paris-Saclay, CNRS, LMPS - Laboratoire de Mécanique Paris-Saclay, 91190, Gif-sur-Yvette, France}
\email{boris.kolev@math.cnrs.fr}

\date{December 19, 2025}%

\begin{abstract}

  This article develops a unified variational framework for configurational (or material) forces in both Classical (3D, non-relativistic) and Relativistic (4D) Continuum Mechanics. Configurational forces describe the evolution of material defects—such as cracks, dislocations, and interfaces—which move relative to the material rather than through physical space. In the classical setting of hyperelasticity, the authors revisit the balance of configurational forces using an intrinsic Lagrangian formulation, where the material body is modeled as an abstract three-dimensional manifold. By treating the reference configuration as a variable and performing a Lagrangian variation with respect to it, they show that the configurational forces balance naturally emerges. Importantly, this balance equation is not independent: it is equivalent to the standard balance of linear momentum combined with constitutive relations, and it is expressed through the Eshelby stress tensor on the reference configuration. The framework is then extended to Relativistic Hyperelasticity within General Relativity. Matter is described by a matter field, a vector valued function, defined on the four-dimensional Universe, and the Lagrangian (\emph{i.e.}, Action) includes both matter and gravitational contributions. Two stress–energy tensors arise: the Noether stress–energy tensor (from variations with respect to the matter field) and the Hilbert stress–energy tensor (from variations with respect to the Universe metric). Assuming General Covariance, the authors prove that these tensors and their associated balance laws are equivalent.
  By introducing the notion of an observer and specializing to static spacetimes, the authors define a relativistic generalization of the deformation and derive a four-dimensional Eshelby tensor. They show that in Special Relativity, as in Classical Continuum Mechanics, the relativistic configurational forces balance is not a new equation but follows from the conservation laws of the Noether stress–energy tensor. Finally, they recover the classical configurational forces balance as the non-relativistic limit of the relativistic theory.
  Overall, the paper provides a rigorous geometric and variational interpretation of configurational forces, unifying classical and relativistic formulations and clarifying their deep connection with standard equilibrium equations.

\end{abstract}

\maketitle

\setcounter{tocdepth}{2}
\tableofcontents

\section*{Introduction}

Configurational forces, also referred to as material forces, emerge in problems involving the evolution of defects within a material. Unlike standard (Newtonian or mechanical) forces, which govern the motion of material particles in physical space, configurational forces govern the motion of entities that migrate relative to the material itself. Typical examples include dislocations, cracks, inclusions, voids, vacancies, and evolving interfaces.
The notion of configurational forces was first introduced in the framework of elasticity and Continuum Mechanics by Eshelby~\cite{Esh1970,Esh1975}. Since then, several theoretical frameworks have been developed  to clarify their physical interpretation and to establish the corresponding balance equations. Historically, one may cite the pull-back approach \cite{Mau1993,KM2001,Mau2002,Mau2002a}, in which configurational force balance is viewed as the projection of mechanical force balance equations onto the material manifold. In this context, configurational forces are closely linked to material uniformity and homogeneity \cite{TN1965,Nol1967} and are interpreted as the driving forces associated with continuous distributions of inhomogeneities \cite{EM1990a,EM1996,Eps2002,ME1998}.
Another notable framework is the basic primitive objects approach proposed by Gurtin \cite{Gur2000,Gur1999,KD2002}, where configurational forces are postulated as fundamental physical quantities, independent of mechanical forces, and their balance laws are obtained through invariance principles. The Noether's theorem approach \cite{KS1971,LMOW2004,MH1994} interprets configurational forces as the conserved quantities associated with material translational invariance of the Lagrangian density. The inverse motion approach \cite{TM+1991,Mau1993,Pod2001,Ste2002} derives configurational force balance equations from the stationarity of the energy (or Action) functional with respect to variations of the reference configuration while keeping the current configuration fixed. Closely related to these ideas is the variational approach \cite{KAS2004,LMOW2004}, which introduces, in addition to the material and spatial configurations, an auxiliary configuration—often called the parameter configuration—that serves as a fixed reference for the motion of defects relative to the material manifold, in analogy with the role of the material configuration for the motion of particles in space \cite{MR1999,YMO2006}.

In classical, \emph{i.e.}, three-dimensional and non-relativistic hyperelastodynamics, the so-called configurational forces balance is not a new conservation equation. By pullback onto the reference configuration $\Omega_{0}$, for example by means of Piola and Ericksen identities  \cite{Eri1977,Rog1977,MT1992}, the configurational forces balance is a reformulation  of the  balance of linear momentum combined with the (hyperelasticity) constitutive equations. This combination implies  that both kinematic energy and strain energy densities enters the final expression of the Eshelby stress tensor.

As mentioned,  there are several ways to derive the configurational forces balance \cite{MT1992,Mau1993,Gur1995,KM2001,SSD2009}. We revisit in \autoref{sec:classical-elastodynamics} the classical variational approaches, and provide a rigorous and geometric interpretation of the so-called horizontal variation introduced in \cite{ZOM2008}. We use the Intrinsic Lagrangian Formulation of Classical Continuum Mechanics by Noll \cite{Nol1972,Nol1978} and Rougée \cite{Rou1980,Rou1991a,Rou1991,Rou2006}, in which  the body $\body$ labeling the material points is a three-dimensional compact and orientable manifold with boundary equipped with a mass measure \cite{KD2021,KD2024}. The  reference configuration is not identified with the body, it is a time independent embedding of $\body$ into the affine Euclidean space $\espace$, playing a key role in the variational derivation of the three-dimensional configurational forces balance (as understood in \cite{YMO2006}). The classical ---non-relativistic--- four-dimensional Eshelby--Noether tensor is finally recalled.

Since the pioneering work of Nordström in 1916 \cite{Nor1916}, truly four-dimensional ---relativistic--- formulations of Mechanics of Solid Materials \cite{LC1985} exist, in physics literature  \cite{Sou1958,Syn1959,Sou1964,Ben1965,KM1992,KM1997,BS2003,Wer2006}, and in mechanics literature~\cite{GE1966,Mau1978a,Mau1978b,Mau1978c,EBT2006,PR2013,PRA2015,NWP2022}.
In the present work, we place ourselves within the appealing Variational Relativity framework developed by Souriau for perfect (hyperelastic) matter \cite{Sou1958,Sou1960,Sou1964} (see also~\cite{KD2024}), with the Gauge Theory mindset~\cite{Ble1981}. The Quantum Mechanics wave function $\psi$ is replaced by a matter field defined as a vector valued function $\Psi\colon \mM \to \body \subset \RR^{3}$, where the Universe $\mM$ is a four-dimensional pseudo-Riemannian manifold  and where one recognizes the three-dimensional body $\body$ labeling the material points of Continuum Mechanics. As proposed by Hilbert \cite{Hil1915}, a General Covariant matter Lagrangian (\emph{i.e.}, Action) is simply added to the Hilbert--Einstein functional for gravitation, as recalled in \autoref{sec:relativistic-elastodynamics}.
The variational extremization of the Lagrangian defines two stress--energy tensors (more precisely stress--impulsion--energy tensors):
\begin{itemize}
  \item the Hilbert stress--energy tensor $\bTH$ \cite{Hil1915}, a contravariant fourth-order tensor on the Universe $\mM$, defined from the variational derivative of the Lagrangian with respect to the (Lorentzian) Universe metric $g$,
  \item the Noether  stress--energy tensor $\bTN{\Psi}$ \cite{Noe1918}, a mixed fourth-order tensor on the Universe $\mM$,  defined from the variational derivative of the Lagrangian with respect to the matter field~$\Psi$.
\end{itemize}
At the end of  \autoref{sec:relativistic-elastodynamics}, we show that to assume that the matter Lagrangian is General Covariant links the two stress--energy tensors, as $\bTN{\Psi}=g\bTH$, and makes the associated Euler--Lagrange equations and conservation laws equivalent.

This General Relativity framework does not presuppose the existence of time, and then the not so direct introduction of a spacetime is interpreted as the definition of an observer. For Souriau, an observer is a particular (passive) matter field \cite[Section 39]{Sou1964}. In \autoref{sec:static-spacetime}, we expose and generalize his construction to possibly non flat but static spacetimes. This allows us to define a three-dimensional quotient manifold $\espace$, that becomes the affine Euclidean space in case of the flat Minkowski metric $g=\eta$ of Special Relativity. An \emph{effective matter field} is accordingly defined, with values in a reference configuration system $\Omega_{0}\subset \espace$, rather than in $\body$. A four-dimensional relativistic generalization of the Continuum Mechanics deformation, $\Phi$, is then proposed. \autoref{sec:relativistic-configurational-forces} shows that its related Noether tensor $\bTN{\Phi}$ is a relativistic generalization of the well-known  (non-relativistic) spacetime Eshelby tensor, providing a direct formulation (on the reference configuration) of the configurational forces balance. We prove that in Special Relativity, as in 3D Classical Continuum Mechanics, the corresponding balance equation is  not a new equation: it is implied by the relativistic conservation laws for the Noether stress-energy tensor $\bTN{\Psi}$. Finally, in \autoref{sec:classical-limit}, we provide the classical limits of the relativistic balance laws, completing then the limits obtained in \cite{KD2023} by the classical limit of the configurational forces balance.

We thus propose a novel variational framework inspired by General Relativity to derive the balance equations of configurational forces within a rigorous and natural mathematical setting, viewed as an appropriate limiting case of General Relativity.

\subsection*{Notations.}

We denote by $g( \cdot , \cdot )$ the scalar product for a metric $g$, by $\varphi^{*}$ (resp. $\varphi_{*}$) the pullback (resp. pushforward) operation by $\varphi$; $(\cdot)^{\star}$ stands for the dual transpose, the identity operator on the tangent space $TM$ of an $n$-dimensional manifold is noted $\bI_{n}$, and its dual transpose on the cotangent space $T^{\star}M$, by $\bI_{n}^{\star}$.

\section{Classical variational elastodynamics}
\label{sec:classical-elastodynamics}

In Classical Continuum Mechanics, the ``space'' $\espace$ is an affine Euclidean space of dimension 3 and the Euclidean metric is denoted by $\bq$. The corresponding Riemannian volume form is then denoted by $\vol_{\bq}$ and written as
\begin{equation*}
  \sqrt{\det(q_{ij})} \; \rd x^{1} \wedge \rd x^{2} \wedge \rd x^{3},
\end{equation*}
in any given system of local coordinates $(x^{i})$. The material medium is parameterized by a 3D compact and orientable manifold with boundary, noted $\body$, and called the \emph{body}. This manifold $\body$ is equipped with a volume form, \textit{i.e.}, a nowhere vanishing 3-form, the \emph{mass measure} $\mu$~\cite{TN1965}.

A \emph{configuration} of a material medium is represented by a smooth, orientation-preserving \emph{embedding}~\cite{Nol1972} (particles cannot occupy the same point in space)
\begin{equation*}
  \pp : \body \to \espace, \qquad \bX \mapsto \xx,
\end{equation*}
also referred to, sometimes, as a \emph{placement} in mechanics \cite{Rou1980,Rou1991a}. Its linear tangent map
\begin{equation*}
  T\pp: T\body \to T\espace
\end{equation*}
is denoted by~$\bF$ and the submanifold $\Omega_{\pp} = \pp(\body)$ of $\espace$ corresponds to a \emph{configuration system}. It is common to fix a \emph{reference configuration} $\pp_{0}$ and thus to substitute the body $\body$ by the \emph{submanifold $\Omega_{0} = \pp_{0}(\body)$ of the ambient space $\espace$}, a reference configuration system. In that case the \emph{deformation} $ \phi := \pp \circ \pp_{0}^{-1}$ (sometimes called the transformation),
\begin{equation*}
  \phi \colon \Omega_{0} \to \Omega_{\pp}, \qquad \xx_{0} \mapsto \xx,
\end{equation*}
is used rather than the placement $\pp$. Its linear tangent map $T\phi: T\Omega_{0} \to T\Omega$ is traditionally referred to as the \emph{deformation gradient} (or as the transformation gradient). It is here noted $\bF_{\phi}$. Setting $\bF_{0}:=T\pp_{0}$, we have
\begin{equation*}
  \bF_{\phi} = T\phi = T\pp .(T\pp_{0})^{-1} =\bF \bF_{0}^{-1}.
\end{equation*}

\begin{rem}
  The mass density $\rho$ is defined implicitly by
  \begin{equation*}
    \pp_{*} \, \mu = \rho\, \vol_{\bq}, \qquad \rho \in \Cinf(\Omega_{\pp}, \RR).
  \end{equation*}
  Given a reference configuration $\pp_{0}$, one defines similarly the \emph{reference mass density} $\rho_{0}$ by the implicit relation
  \begin{equation}\label{eq:MCrho0mu}
    (\pp_{0})_{*} \,\mu = \rho_{0}\, \vol_{\bq},\qquad \rho \in \Cinf(\Omega_{0}, \RR).
  \end{equation}
  The \emph{law of conservation of mass} $\rho_{0} /(\rho \circ \phi)= \det \bF_{\phi}$ is deduced from
  \begin{equation}\label{eq:MC}
    \phi^{*}  (\rho \vol_{\bq},)=\rho_{0} \,\vol_{\bq} = J_{\phi} (\rho \circ \phi) \,\vol_{\bq}, \qquad J_{\phi} := \det \bF_{\phi}.
  \end{equation}
\end{rem}

The \emph{configuration space} in Continuum Mechanics is thus the set, denoted by $\Emb(\body,\espace)$, of smooth embeddings of $\body$ into $\espace$. This set can be endowed with a differential manifold structure of infinite dimension, indeed, an open set of the Fréchet affine space $\Cinf(\body,\espace)$~\cite{Ham1982,Mil1984}. The tangent space to $\Emb(\body,\espace)$ at a point $\pp \in \Emb(\body,\espace)$ is described as follows. Let $\pp(t)$ be a smooth curve in $\Emb(\body,\espace)$ (a path of embeddings) such that $\pp(0) = \pp$, then $(\partial_{t}\pp)(0) = \delta\pp$ is a tangent vector at $\pp$ (or a variation of $\pp$) and can be interpreted as a \emph{virtual Lagrangian displacement}~\cite{ES1980}. The choice of a reference configuration $\pp_{0}$ induces a diffeomorphism (or \emph{change of parametrization} of the configuration space)
\begin{equation*}
  \Emb(\body,\espace) \longrightarrow \Emb(\Omega_{0},\espace), \qquad \pp \mapsto \phi := \pp \circ {\pp_{0}}^{-1}, \quad \text{where} \quad \Omega_{0} \subset \espace.
\end{equation*}


\subsection{Intrinsic Lagrangian}

\emph{Lagrangian Mechanics} was initially formulated by Lagrange, first for finite discrete systems (see~\cite{Sou1997} for an excellent book on that topic). It consists in writing down a functional, the \emph{Lagrangian}, defined on paths $\breve{\kappa} \colon [t_{0},t_{1}] \to \mathcal{C}$, where $\mathcal{C}$ is the configuration space of the mechanical system. In the case of first gradient hyperelasticity, the Lagrangian (also called the Action) is written as
\begin{equation}\label{eq:3D-Lagrangian}
  \mL[\breve{\pp}] = \int_{t_{0}}^{t_{1}} \left(\int_{\body} \ell\left(\bX, p(t)(\bX), \partial_{t}\pp(t)(\bX), \bF(t)(\bX)\right) \mu \right) \rd t,
\end{equation}
where the configuration space is here the set $\mathcal{C}:= \Emb(\body,\espace)$ of smooth embeddings from $\body$ to $\espace$ and the path of configurations $\breve\kappa$ is $\breve{\pp}=(\pp(t))$,  or simply $\pp=(\pp(t))$ with a notation abuse. The Lagrangian density $\ell$ is here a specific energy. It depends \emph{a priori} on the coordinates on the body, $\bX$, and on the first jet of $\pp$, at least for \emph{local systems}. Note that conservative body forces, but also boundary terms such as prescribed pressure (when non holonomic conditions are satisfied \cite{Sew1965,Sew1967,Bea1970,PC1991,KD2021}), can be included in the Lagrangian \cite{Bal1976/77}.

\begin{rem}
  We omit the explicit dependency on time $t$ of the Lagrangian density \cite{MH1994}. We will show in \autoref{sec:relativistic-elastodynamics} that this feature is a consequence of general covariance and the assumption of flatness of spacetime.
\end{rem}

Introducing a reference configuration $\pp_{0} \in \Emb(\body,\espace)$ and the deformation path $\breve{\phi}=(\phi(t))$  parameterized by time, or simply $\phi=(\phi(t))$,
\begin{equation}\label{eq:phi3D}
  \phi(t) = \pp(t) \circ \pp_{0}^{-1},
  \qquad \Omega_{0} \to \Omega_{\pp(t)},
\end{equation}
such that
\begin{equation*}
  \partial_{t} \pp(t) = \partial_{t} \phi (t) \circ \pp_{0}
  \qquad\text{and}\qquad
  \bF(t)= T\phi (t).T \pp_{0} = \bF_{\phi}(t).\bF_{0},
\end{equation*}
one can recast this integral on the reference configuration. We get
\begin{align*}
  \mL[\pp] & = \int_{t_{0}}^{t_{1}} \left( \int_{\Omega_{0}} (\pp_{0})_{*} \big[\ell\left(\bX, p(t)(\bX), \partial_{t}\pp(t)(\bX), \bF_{\phi}(t).\bF_{0}(\bX)\right) \mu\big] \right) \rd t                                        \\
           & = \int_{t_{0}}^{t_{1}} \left(  \int_{\Omega_{0}} \ell\left(\pp_{0}^{-1}(\xx_{0}), \phi(t)(\xx_{0}),\partial_{t}\phi(t)(\xx_{0}), \bF_{\phi}(t).\bF_{0}(\pp_{0}^{-1}(\xx_{0}))\right)\rho_{0}\vol_{\bq}\right) \rd t ,
\end{align*}
where the second equality uses mass conservation \eqref{eq:MCrho0mu}.
Usually, one prefers to express the Lagrangian density on the reference configuration, as a function of the deformation $\phi$ and of its first jets $\partial_{t}\phi$ and $T\phi=\bF_{\phi}$, setting
\begin{multline}\label{eq:L03D}
  L_{0}\big(\xx_{0},\phi(t,\xx_{0}), \partial_{t}\phi(t,\xx_{0}),\bF_{\phi}(t,\xx_{0})\big)
  \\
  := \rho_{0}\ell\big(\pp_{0}^{-1}(\xx_{0}), \phi(t)(\xx_{0}),\partial_{t}\phi(t)(\xx_{0}), \bF_{\phi}(t)\bF_{0}(\pp_{0}^{-1}(\xx_{0}))\big) ,
\end{multline}
so that
\begin{equation*}
  \mL[\pp] = \int_{t_{0}}^{t_{1}} \left(\int_{\Omega_{0}}  L_{0}(\xx_{0},\phi(t,\xx_{0}), \partial_{t}\phi(t,\xx_{0}),\bF_{\phi}(t,\xx_{0})) \,\vol_{\bq}\right) \rd t.
\end{equation*}
For the sake of conciseness, this expression can be simply written in condensed form
\begin{equation}\label{eq:LCCM}
  \mL[\pp] = \int_{t_{0}}^{t_{1}} \left(\int_{\Omega_{0}}  L_{0}(\xx_{0},\phi, \partial_{t}\phi,\bF_{\phi}) \,\vol_{\bq}\right) \rd t,
\end{equation}
meaning that the constitutive tensors are defined on the reference configuration.

\begin{rem}\label{rem:noF0}
  Let us emphasize that  in definition \eqref{eq:L03D} the linear tangent map of the reference embedding (\emph{i.e.}, the reference deformation gradient),
  \begin{equation*}
    \bF_{0}=T\pp_{0}\colon T\body \to T\Omega_{0},
  \end{equation*}
  is omitted from the left inner arguments of the Lagrangian density $L_{0}$. This feature corresponds to the mechanical choice of defining the material properties through a possibly heterogeneous mass density $\rho_{0}=\rho_{0}(\xx_{0})$ and through constitutive tensors fields $\bA=\bA(\xx_{0})$ on the reference configuration  $\Omega_{0}=\pp_{0}(\body)$.  Three such examples are: $(i)$ the kinematic energy density
  \begin{equation*}
    K\big(\xx_{0}, \partial_{t}\phi\big) = \frac{1}{2} \rho_{0} \|\VV\|^{2},
  \end{equation*}
  where $\VV=\partial_{t}\phi$ is the Lagrangian velocity,
  $(ii)$ the potential energy density $\rho_{0}\, \bg \cdot \phi$,
  introducing the gravity $\bg$,
  and $(iii)$ the (anisotropic) Mooney energy density,
  \begin{equation}\label{eq:Mooney}
    W\big(\xx_{0},\bF_{\phi}\big)
    = \tr \left( \bA \bC\right)-\tr \left( \bA \bq\right),
    \quad
    \text{on} \quad \Omega_{0},
  \end{equation}
  where $\bq$ is the Euclidean metric, $\bC:= \phi^{*} \bq = \bF_{\phi}^{\star}\, \bq\, \bF_{\phi}$ is the right Cauchy--Green tensor, and $\bA$ is a contravariant second order constitutive tensor. Note that one has $\bA=A\, \bq^{-1}$  in the isotropic case, where $A$ is the Mooney (material) parameter~\cite{Moo1940}.
\end{rem}


\subsection{Natural variations of the Lagrangian}

In formula \eqref{eq:LCCM}, it seems forgotten that the definition of the deformation $\phi$ requires the \emph{choice of a reference configuration} $\pp_{0}$. For a natural derivation of the configurational forces balance, we prefer to rewrite the Lagrangian so that the reference configuration $\pp_{0}$ acts as a parameter. The dependency on $\pp_{0}$ is made explicit thanks to the pullback of the space integral onto the body $\body$, as
\begin{equation}\label{eq:Lphip0}
  \begin{aligned}
    \mL[\pp] = \mL_{0}[\phi;\pp_{0}]
     & = \int_{t_{0}}^{t_{1}} \left(\int_{\body}  L_{0}\big(\pp_{0}(\bX),\phi(t,\pp_{0}(\bX)), \partial_{t}\phi(t,\pp_{0}(\bX)),\bF_{\phi}(t,\pp_{0}(\bX))\big) \,{\pp_{0}}^{\! \! \! *}\,\vol_{\bq}\right) \rd t,
  \end{aligned}
\end{equation}
or simply as
\begin{equation}\label{eq:REFLAG}
  \mL[\pp] = \mL_{0}[\phi;\pp_{0}] = \int_{t_{0}}^{t_{1}} \left(\int_{\body}  L_{0}\big(\pp_{0},\phi\circ \pp_{0}, \partial_{t}\phi\circ \pp_{0},\bF_{\phi}\circ \pp_{0}\big) \,{\pp_{0}}^{\! \! \! *}\,\vol_{\bq}\right) \rd t,
\end{equation}
where the constitutive tensors are implicitly defined on the reference configuration. From the previous definition of the Lagrangian \eqref{eq:REFLAG}, we can now introduce the corresponding variational principle by the extremization of the Lagrangian:
\begin{equation}\label{eq:VARPLE}
  \delta \mL[\pp] = \delta \mL_{0}[\phi;\pp_{0}] = \delta_{\phi}\mL_{0} + \delta_{\pp_{0}}\mL_{0} = D_{\phi}\mL_{0}.\delta \phi + D_{\pp_{0}}\mL_{0}.\delta \pp_{0} = 0 \qquad \forall \delta{\phi}, \delta \pp_{0}
\end{equation}
where $D_{\phi}\mL_{0}.\delta \phi$ and $D_{\pp_{0}}\mL_{0}.\delta \pp_{0}$ are the variational derivatives with respect to $\phi$ and $\pp_{0}$. Let us denote by $\langle \cdot, \cdot \rangle$ the scalar product for the Euclidean metric $\bq$ on the Euclidean space~$\espace$. We get then the following variations of the Lagrangian:
\begin{itemize}
  \item the variation $\delta_{\phi}\mL_{0}$ of $\mL_{0}[\phi;\pp_{0}]$ relatively to the path of deformation  $\phi=(\phi(t))$ is \cite{MH1994}
        \begin{multline}\label{eq:varphi}
          \delta_{\phi}\mL_{0} =D_{\phi}\mL_{0}. \delta \phi = \int_{t_{0}}^{t_{1}} \left( \int_{\Omega_{0}} \left(\pd{L_{0}}{\phi} - \dive\left(\pd{L_{0}}{\bF_{\phi}}\right) - \dd{}{t}\left(\pd{L_{0}}{(\partial_{t}\phi)}\right) \right)\cdot\delta \phi \, \vol_{\bq}\right) \rd t
          \\
          + \int_{t_{0}}^{t_{1}} \left(\int_{\partial\Omega_{0}} \Big\langle\delta \phi\cdot\pd{L_{0}}{\bF_{\phi}} ,\nn_{0}\Big\rangle \, \rd \ba_{0} \right) \rd t \,  + \, \int_{\Omega_{0}} \left[\pd{L_{0}}{(\partial_{t}\phi)} \cdot \delta \phi\right]_{t_{0}}^{t_{1}}\vol_{\bq} ,
        \end{multline}

  \item the variation $\delta_{\pp_{0}}\mL_{0}$ of $\mL_{0}[\phi;\pp_{0}]$ relatively to the reference configuration $\pp_{0}$ is
        \begin{multline}\label{eq:p0}
          \delta_{\pp_{0}}\mL_{0} = D_{\pp_{0}}\mL_{0}.\delta \pp_{0} =
          \\
          \int_{t_{0}}^{t_{1}} \left( \int_{\Omega_{0}} \left(\pd{L_{0}}{\xx_{0}} + \pd{L_{0}}{\phi} \cdot \bF_{\phi} + \pd{L_{0}}{\bF_{\phi}} :\rdx\bF_{\phi} + \pd{L_{0}}{(\partial_{t}\phi)} \cdot (\partial_{t}\bF_{\phi}) - \dive (L_{0}\,\Idt^{\star}) \right) \cdot \delta \pp_{0} \circ {\pp_{0}}^{-1} \, \vol_{\bq}\right)\rd t
          \\
          + \int_{t_{0}}^{t_{1}} \left(\int_{\partial\Omega_{0}}  L_{0} \, \big\langle \delta \pp_{0} \circ {\pp_{0}}^{-1} , \nn_{0} \big\rangle\, \rd \ba_{0} \right) \rd t .
        \end{multline}
\end{itemize}
The details of the calculations for the variation relative to $\phi$ and $\pp_{0}$ are given in \autoref{sec:calculation-details}.

Formally speaking, the proposed approach (already introduced in~\cite{YMO2006}) provides a rigorous and geometric interpretation of the so-called horizontal variations introduced in~\cite{ZOM2008} and which leads to the same result. A key point of our approach it that it clarifies the underlining assumption ---fundamental for the further derivation of the configurational forces balance--- that the variation with respect to the reference configuration is performed at fixed constitutive tensors, defined on the reference configuration system $\Omega_{0}$ (see Remark~\ref{rem:noF0}).


\subsection{Euler’s Law of the balance of linear momentum}
\label{subsec:Euler-law}

Applying~\eqref{eq:varphi} for variations $\delta \phi$ which vanishes on the boundary $\partial \Omega_{0}$ and for times $t = t_{0}$ and $t=t_{1}$, one gets then that a path of deformation $\phi$ is a critical point of the Lagrangian $\mL$ if and only if it satisfies the \emph{Euler--Lagrange equation}
\begin{equation}\label{eq:MMC-Euler--Lagrange}
  \pd{L_{0}}{\phi} - \dive\left(\pd{L_{0}}{\bF_{\phi}}\right) - \dd{}{t}\left(\pd{L_{0}}{(\partial_{t}\phi)}\right) = 0, \qquad \text{on} \quad \Omega_{0}.
\end{equation}
where the time derivative $\displaystyle \dd{f}{t}$ on $\Omega_{0}$ stands for the total derivative $\displaystyle \dd{}{t} f(t, \xx_{0}, \partial_{t} \phi(t, \xx_{0}), \bF_{\phi}(t, \xx_{0}))$, with respect to time $t$.

In order to avoid cumbersome notations and to clarify its interpretation with mechanical contents, we introduce the following quantities:
\begin{align*}
  \ff      & := \pd{L_{0}}{\phi},                                 &  & \text{(Exterior force density)},
  \\
  \bp      & := \pd{L_{0}}{(\partial_{t}\phi)} = \pd{L_{0}}{\VV}, &  & \text{(Linear momentum density)},
  \\
  \hat \bP & := - \pd{L_{0}}{T\phi} = - \pd{L_{0}}{\bF_{\phi}},   &  & \text{(First Piola-Kirchhoff stress tensor)},
  \\
  \VV      & := \partial_{t} \phi,                                &  & \text{(Lagrangian velocity)}.
\end{align*}
of components $f_{i}$, $p_{i}$, ${P_{i}}^{J}$ and $V^{i}$, respectively. The Euler--Lagrange equation~\eqref{eq:MMC-Euler--Lagrange} recasts then as the  balance of linear momentum of Classical Continuum Mechanics,
\begin{equation}\label{eq:MMC-Euler--Lagrange-bis}
  \dive \hat \bP +\ff = \dd{\bp}{t} , \qquad \text{on} \quad \Omega_{0}.
\end{equation}

The total energy density $E_{0}=E_{0}(t, \xx_{0})$ of the system on the reference configuration $\Omega_{0}$, is defined as the Legendre transform of the Lagrangian $L_{0}$. It is written as
\begin{equation}\label{eq:ECCM}
  E_{0}:= \pd{L_{0}}{(\partial_{t}\phi)} \cdot \partial_{t}\phi - L_{0} = \bp \cdot \VV - L_{0}.
\end{equation}
We have thus
\begin{align*}
  \dd{E_{0}}{t} & = \dd{}{t} \left(\pd{L_{0}}{(\partial_{t}\phi)}\right)\cdot \partial_{t}\phi + \left(\pd{L_{0}}{(\partial_{t}\phi)}\right)
  \cdot {\partial_{t}}^{2}\phi - \dd{L_{0}}{t}
  \\
                & = \left\{ \dd{}{t}\left(\pd{L_{0}}{(\partial_{t}\phi)}\right) - \left(\pd{L_{0}}{\phi}\right)\right\} \cdot \partial_{t}\phi
  - \left(\pd{L_{0}}{\bF_{\phi}}\right) : \partial_{t} \rdx \phi - \pd{L_{0}}{t}.
\end{align*}
Hence, along a critical path $\phi$ (which satisfies the Euler-Lagrange equations~\eqref{eq:MMC-Euler--Lagrange}), we get
\begin{equation*}
  \dd{E_{0}}{t}  = - \dive\left(\pd{L_{0}}{\bF_{\phi}}\right)\cdot\partial_{t}\phi
  - \left(\pd{L_{0}}{\bF_{\phi}}\right) :\partial_{t} \rdx \phi - \pd{L_{0}}{t}.
\end{equation*}
Now, using the equality
\begin{equation*}
  \dive\left(\partial_{t}\phi \cdot \pd{L_{0}}{\bF_{\phi}}\right) = \dive\left(\pd{L_{0}}{\bF_{\phi}}\right)\cdot \partial_{t}\phi + \left(\pd{L_{0}}{\bF_{\phi}}\right) : \partial_{t} \rdx \phi,
\end{equation*}
one gets
\begin{equation}\label{eq:energy-density-balance}
  \dd{E_{0}}{t} + \dive \left( \pd{L_{0}}{\bF_{\phi}} \cdot \partial_{t}\phi \right) + \frac{\partial L_{0}}{\partial t} = 0,
  \qquad
  \text{on $\Omega_{0}$},
\end{equation}
which recasts as the Classical Continuum Mechanics balance for the \emph{total energy density}
\begin{equation}\label{eq:energy-conservation}
  \dd{E_{0}}{t} = \dive\big(\VV \hat \bP \big),
\end{equation}
when the Lagrangian density $L_{0}$ does not depend explicitly on time $t$.

\begin{rem}
  Here, we can recall the straightforward consequences of the well-known Noether's theorem~\cite{Noe1918} for Classical Continuum Mechanics~\cite[Section 5.5]{MH1994} (see also~\cite{MW2001,LMOW2003}). This allows to recover the balance of linear momentum \eqref{eq:MMC-Euler--Lagrange-bis} from the spatial translation invariance of the Lagrangian. Similarly, one recovers the balance of angular momentum (\emph{i.e.}, the symmetry of the Cauchy stress tensor), from the invariance by rotation of the Lagrangian.
\end{rem}


\subsection{Recovering the balance equation of Configurational forces}
\label{subsec:classical-configurational-forces}

There are several ways to derive the balance of configurational forces in the literature~\cite{MT1992,Mau1993,Gur1995,KM2001,SSD2009}. Here, we propose a new one, in the framework of the \emph{Intrinsic Lagrangian Formulation of Continuum Mechanics} of Noll~\cite{Nol1972,Nol1978} and Rougée \cite{Rou1980,Rou1991a,Rou1991,Rou2006}: we derive it simply as a critical point of the Lagrangian $\mL_{0}[\phi;\pp_{0}]$ defined in~\eqref{eq:3D-Lagrangian} (expressed on the body $\body$ and in which the reference configuration $\pp_{0}$ acts explicitly as a parameter). Such a critical point is a couple deformation path/reference configuration $(\phi,\pp_{0})$ which satisfies the vanishing of both variations
\begin{equation*}
  \delta_{\phi}\mL_{0} = 0, \quad \text{and} \quad \delta_{\pp_{0}}\mL_{0} = 0.
\end{equation*}
The first equation leads back to Euler--Lagrange equation~\eqref{eq:MMC-Euler--Lagrange} and to the equilibrium equation \eqref{eq:MMC-Euler--Lagrange-bis}, whereas the second equation, applied to variations $\delta \pp_{0}$ which vanish on the boundary $\partial\body$, gives, by~\eqref{eq:p0},
\begin{equation*}
  \pd{L_{0}}{\xx_{0}} + \pd{L_{0}}{\phi}\cdot \bF_{\phi} + \pd{L_{0}}{\bF_{\phi}} :\rdx\bF_{\phi} + \pd{L_{0}}{(\partial_{t}\phi)} \cdot(\partial_{t}\bF_{\phi}) - \dive (L_{0}\, \Id^{\star}) = 0.
\end{equation*}
By the formula
\begin{equation*}
  \pd{L_{0}}{\bF_{\phi}} :\rdx\bF_{\phi} = \dive \left(\bF_{\phi}^{\star}\cdot \pd{L_{0}}{\bF_{\phi}}\right) - \dive \left(\pd{L_{0}}{\bF_{\phi}}\right)\cdot \bF_{\phi},
\end{equation*}
which details, in components, as
\begin{equation*}
  {\left(\pd{L_{0}}{\bF_{\phi}}\right)_{i}}^{J}{(\bF_{\phi})^{i}}_{J,K} = {\left(\pd{L_{0}}{\bF_{\phi}}\right)_{i}}^{J}{(\bF_{\phi})^{i}}_{K,J}
  = \left[{\left(\pd{L_{0}}{\bF_{\phi}}\right)_{i}}^{J}{(\bF_{\phi})^{i}}_{K}\right]_{,J}
  -{{\left(\pd{L_{0}}{\bF_{\phi}}\right)_{i}}^{J}}_{,J}{(\bF_{\phi})^{i}}_{K} ,
\end{equation*}
we get
\begin{equation*}
  \pd{L_{0}}{(\partial_{t}\phi)} \cdot (\partial_{t}\bF_{\phi}) = \dd{}{t}\left(\pd{L_{0}}{(\partial_{t}\phi)} \cdot\bF_{\phi}\right) - \dd{}{t}\left(\pd{L_{0}}{(\partial_{t}\phi)} \right)\cdot\bF_{\phi}.
\end{equation*}
We have therefore
\begin{multline*}
  \pd{L_{0}}{\xx_{0}} + \pd{L_{0}}{\phi} \cdot \bF_{\phi} + \left\{\dive \left(\bF_{\phi}^{\star} \cdot \pd{L_{0}}{\bF_{\phi}}\right) - \dive \left(\pd{L_{0}}{\bF_{\phi}}\right)\cdot\bF_{\phi}\right\}
  \\
  + \left\{\dd{}{t}\left(\pd{L_{0}}{(\partial_{t}\phi)}\cdot \bF_{\phi}\right) - \dd{}{t}\left(\pd{L_{0}}{(\partial_{t}\phi)} \right)\cdot\bF_{\phi}\right\} - \dive (L_{0}\,\Idt^{\star}) = 0.
\end{multline*}
This last equation recasts as
\begin{multline*}
  \pd{L_{0}}{\xx_{0}} + \left\{\pd{L_{0}}{\phi} - \dive \left(\pd{L_{0}}{\bF_{\phi}}\right) - \dd{}{t}\left(\pd{L_{0}}{(\partial_{t}\phi)} \right)\right\} \cdot\bF_{\phi}
  \\
  + \dive \left(\bF_{\phi}^{\star} \cdot \pd{L_{0}}{\bF_{\phi}} - L_{0}\,\Idt^{\star}\right) + \dd{}{t}\left(\pd{L_{0}}{(\partial_{t}\phi)} \cdot\bF_{\phi}\right) = 0.
\end{multline*}
Making use of Euler--Lagrange equation~\eqref{eq:MMC-Euler--Lagrange}, we obtain finally
\begin{equation}\label{eq:general_config-forces-conservation}
  \pd{L_{0}}{\xx_{0}} + \dive \left(\bF_{\phi}^{\star} \cdot \pd{L_{0}}{\bF_{\phi}} - L_{0}\,\Idt^{\star}\right) + \dd{}{t}\left(\pd{L_{0}}{(\partial_{t}\phi)} \cdot \bF_{\phi}\right) = 0,
\end{equation}
where we have introduced the first Piola--Kirchhoff tensor $\hat \bP =- {\partial L_{0}}/{\partial \bF_{\phi}}$ (with components
${P}_{i}{}^{J}=-{\partial L_{0}}/{\partial {F^{\;i}_{\phi\,J}}}$), and the linear momentum density covector $\bp$. We end up with the balance equation
\begin{equation}\label{eq:config-forces-conservation}
  \pd{L_{0}}{\xx_{0}} - \dive \left(\bF_{\phi}^{\star} \cdot \hat \bP+ L_{0}\,\Idt^{\star}\right) + \dd{}{t}\left(\bp\cdot \bF_{\phi}\right) = 0,
\end{equation}
which is the \emph{configurational forces} balance of Classical Continuum Mechanics on $\Omega_{0}$. The second order tensor
\begin{equation}\label{eq:B3D-classical}
  \hat \bB:=-L_{0}\,\Idt^{\star}-\bF_{\phi}^{\star} \cdot \hat \bP
\end{equation}
is the 3D \emph{Eshelby stress tensor}, a mixed tensor field $\hat \bB \colon T^{\star} \Omega_{0}~\to~T^{\star} \Omega_{0}$, with components $({B}_{I}{}^{J})$, and defined on the reference configuration.

\begin{rem}
  In the particular case of homogeneous Hyperelasticity, the Lagrangian density is given by (see Remark \ref{rem:noF0})
  \begin{equation}\label{eq:L0classic}
    L_{0}(\phi, \VV, \bF_{\phi})=K(\VV)- W(\bF_{\phi}) + \ff\cdot \phi,
    \qquad
    K=\frac{1}{2} \rho_{0}\langle \VV, \VV\rangle,
  \end{equation}
  where $K$ is the kinematic energy density and $W$ is the hyperelastic energy density. The configurational forces balance takes then the common form~\cite{Esh1970,MT1992}
  \begin{equation}\label{eq:usual-config-forces}
    \dive \hat \bB + \dd{}{t}\left(\bp \cdot \bF_{\phi}\right) = 0,
    \quad \text{where} \quad
    \begin{cases}
      \hat \bB:=\left(W-K\right) \Idt^{\star}-\bF_{\phi}^{\star} \cdot \hat \bP,
      \\
      \hat \bP  = \displaystyle \pd{W}{\bF_{\phi}},
      \\
      \bp= \rho_{0} \VV^{\flat},
    \end{cases}
  \end{equation}
  and $\VV^{\flat}=\bq \VV$ is the covariant version of the Lagrangian velocity.
\end{rem}

One can also recover the configurational forces balance using the Noether stress--energy tensor defined in \autoref{sec:variational-calculus} (see also~\cite{MT1992,KM2001}). In the present case, the field variable
\begin{equation*}
  \phi \colon \RR \times \Omega_{0} \to \espace, \qquad ({x_{0}}^{0} = t,{x_{0}}^{I}) \mapsto (x^{i}=\phi^{i}(t, {x_{0}}^{I})),
  \qquad
  1 \le I \le 3 \;\, \text{and} \;\, 1 \le i \le 3,
\end{equation*}
is a submersion and we get the 4D Noether--Eshelby tensor (sometimes referred to as spacetime energy-material momentum tensor or spacetime Eshelby tensor \cite{Esh1975,Mau1992,Gur1995}),
\begin{equation*}
  \bEsh := \bTN{\phi} = (T\phi)^{\star}\cdot \pd{L_{0}}{T\phi}- L_{0} \, \Idq^{\star},
\end{equation*}
with components,
\begin{equation*}
  {\Eshcomp_{\mu}}^{\nu}={(T\phi)^{k}}_{\mu}{\left(\pd{L_{0}}{T\phi}\right)_{k}}^{\nu} - L_{0} \, {\delta_{\mu}}^{\nu}.
\end{equation*}
Since ${(T\phi)^{i}}_{t}=\partial_{t} \phi^{i}$ and ${(T\phi)^{i}}_{J}={(\bF_{\phi})^{i}}_{J}$, we have in particular
\begin{align*}
  {\Eshcomp_{t}}^{t} & = \partial_{t}\phi^{k} \pd{L_{0}}{\partial_{t} \phi^{k}}- L_{0} = E_{0},
  \\
  {\Eshcomp_{t}}^{I} & = \partial_{t}\phi^{k} {\left(\pd{L_{0}}{\bF_{\phi}}\right)_{k}}^{I} = -V^{k} {P_{k}}^{I},
  \\
  {\Eshcomp_{I}}^{t} & = {(\bF_{\phi})^{k}}_{I}\left(\pd{L_{0}}{\partial_{t}\phi^{k}}\right) = p_{k}{(\bF_{\phi})^{k}}_{I},
  \\
  {\Eshcomp_{I}}^{J} & = {(\bF_{\phi})^{k}}_{I} {\left(\pd{L_{0}}{\bF_{\phi}}\right)_{k}}^{J} - L_{0} \, {\delta_{I}}^{J} = -{(\bF_{\phi}^{\star})_{I}}^{k} {P_{k}}^{J} - L_{0} \, {\delta_{I}}^{J} = B_{I}{}^{J}.
\end{align*}
Therefore, the relations of \autoref{sec:variational-calculus}
\begin{equation*}
  \left(\dive  \bEsh\right)_{t} = - \pd{L_{0}}{t} = 0 \quad \text{and} \quad \left(\dive  \bEsh\right)_{I} =  - \pd{L_{0}}{{x_{0}}^{I}}
\end{equation*}
give back both the energy balance \eqref{eq:ECCM} and the configurational forces balance \eqref{eq:usual-config-forces}
\begin{equation}\label{eq:ConfigBalance}
  \dd{E_{0}}{t} = \dive (\VV \cdot \hat \bP), \qquad \frac{\partial L_{0}}{\partial \xx_{0}} + \dive \hat \bB + \dd{}{t} \left(\bp \cdot \bF_{\phi}\right)=0,
\end{equation}
expressed on the reference configuration.

\begin{rem}\label{rem:Esh4Dclassic}
  Equations \eqref{eq:energy-density-balance} and \eqref{eq:general_config-forces-conservation} can also be written in a general spacetime matrix form. Let us introduce the 4D divergence, $\dive^{4D}$, and the block expression of the energy-momentum Eshelby tensor (in Lagrangian coordinate system $(t, \xx_{0})$)
  \begin{equation*}
    \bEsh =
    \begin{bmatrix}
      \displaystyle \frac{\partial L_{0}}{\partial (\partial_{t}\phi)} \cdot \partial_{t}\phi - L_{0} =( {\Eshcomp_{t}}^{t} ) & \displaystyle \partial_{t}\phi  . \displaystyle\pd{L_{0}}{\bF_{\phi}} = ( {\Eshcomp_{t}}^{J} )
      \\
      \bF_{\phi}^{\star} \cdot \displaystyle \pd{L_{0}}{(\partial_{t}\phi)} = ( {\Eshcomp_{I}}^{t} )                          &
      \displaystyle \bF_{\phi}^{\star} \cdot \pd{L_{0}}{\bF_{\phi}} - L_{0}\,\Idt^{\star}= ({\Eshcomp_{I}}^{J})
    \end{bmatrix}.
  \end{equation*}
  This tensor recasts as
  \begin{equation*}
    \bEsh =
    \begin{bmatrix}
      E_{0}                & -\VV \cdot  \hat \bP
      \\
      \bp \cdot \bF_{\phi} & \hat \bB
    \end{bmatrix},
  \end{equation*}
  using the definitions of \autoref{subsec:Euler-law}, and where $\hat \bB$ is the 3D Eshelby stress tensor.
  We have thus
  \begin{equation*}
    \pd{L_{0}}{(t, \xx_{0})} + \dive^{4D} \bEsh =0,
  \end{equation*}
  where
  \begin{equation*}
    \pd{L_{0}}{(t, \xx_{0})}=
    \begin{bmatrix}
      \displaystyle \frac{\partial L_{0}}{\partial t} & \displaystyle\frac{\partial L_{0}}{\partial \xx_{0}}
    \end{bmatrix},
  \end{equation*}
  and
  \begin{equation*}
    \dive^{4D} \bEsh =
    \begin{bmatrix}
      {{\Eshcomp_{t}}^{\nu}}_{, \nu} & {{\Eshcomp_{I}}^{\nu}}_{, \nu}
    \end{bmatrix},
  \end{equation*}
  where Greek indices $\mu, \nu, \dotsc$ run from 0 for time to 3 (where 0 stands for $t$) and Latin indices $I, J, \dotsc$ run from 1 to 3.
\end{rem}

\section{Relativistic variational elastodynamics}
\label{sec:relativistic-elastodynamics}

In General Relativity, the Universe is assumed to be a 4D orientable manifold $\mM$, endowed with an hyperbolic metric $g$, of signature
$(-,+,+,+)$. Its pseudo-Riemannian volume form is denoted by $\vol_{g}$ and written as
\begin{equation}\label{eq:volg}
  \sqrt{-\det (g_{\mu\nu})} \; \rd x^{0} \wedge \rd x^{1} \wedge \rd x^{2} \wedge \rd x^{3},
\end{equation}
in any given system of local coordinates $(x^{\mu})$. Unlike Classical Mechanics, where the spatial space $\espace$ is equipped with a fixed Euclidean metric $\bq$, the metric $g$ is not fixed but is itself a field variable which is influenced by the presence of matter or electromagnetic energy. In the present case, where we are only interested by matter and not by electromagnetism, we represent this ``perfect matter'', like Souriau~\cite{Sou1958,Sou1960,Sou1964,KD2023}, by a smooth vector valued function
\begin{equation*}
  \Psi: \mM \to V,
\end{equation*}
defined on the 4D Universe $\mM$ and with values in a three-dimensional vector space $V$. This 3D vector space encodes possible labels for material points and $\Psi$ is called a \emph{perfect matter field}. We assume therefore that there is a 3D compact orientable manifold with boundary $\body\subset V$, called the \emph{body} (see \autoref{fig:fiberedtimelines}) and endowed with a 3D volume form $\mu\in \Omega^{3}(\body)$, the \emph{mass measure}~\cite{KM1992}. The mass measure on $\body$ encodes the density of matter particles and as in 3D Classical Continuum Mechanics~\cite{TN1965,Nol1972,Nol1974,Nol1978}, the integral
\begin{equation*}
  m=\int_{\body} \mu
\end{equation*}
is interpreted as the total mass of matter under study.

It is further assumed that each perfect matter field $\Psi$ is a \emph{submersion} on $\mW:=\Psi^{-1} (\body)$, meaning that the linear tangent map $T\Psi: T\mW \to TV$ is of rank 3 at each point of $\mW$. Hence, the sets $\Psi^{-1}(\bX)$, $\bX\in \body$, are one-dimensional submanifolds of $\mM$ called \emph{particle World lines}. The set $\mW\subset \mM$ is fibered by these particle World lines and referred to as the \emph{matter World tube}.

\begin{figure}[ht]
  \centering
  \includegraphics[width=11cm]{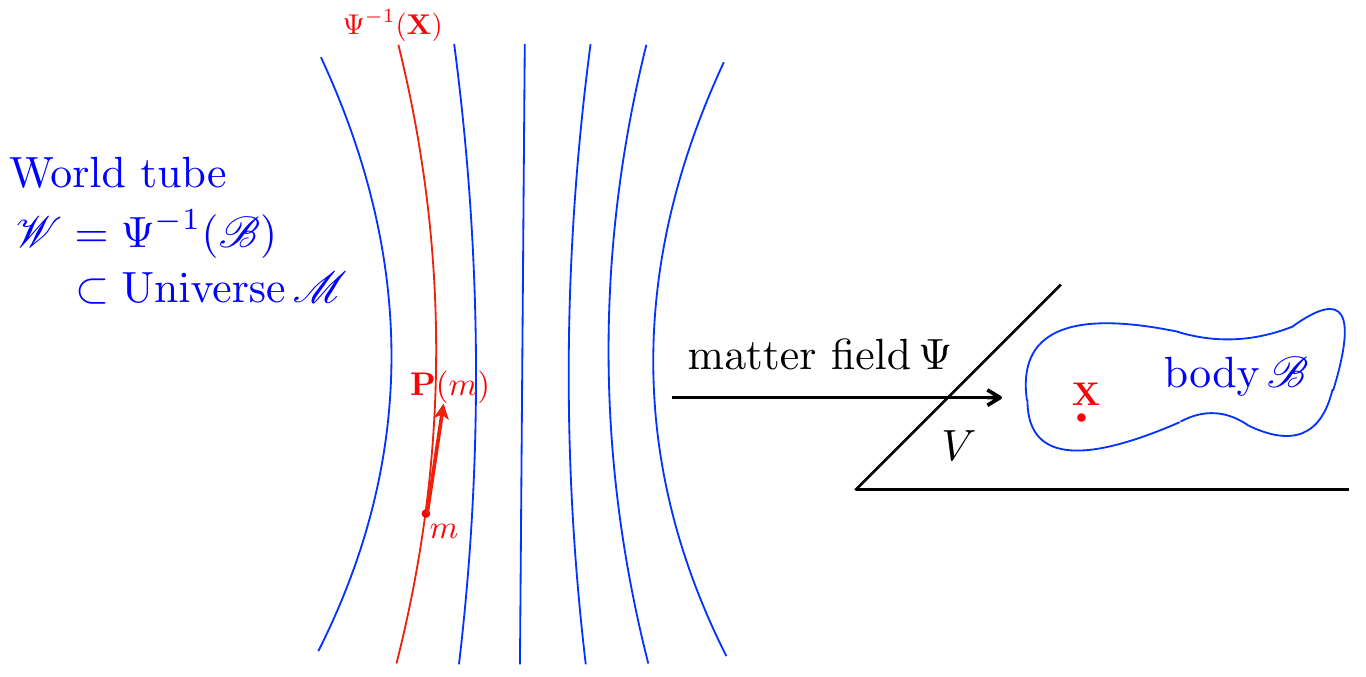}
  \caption{The World tube $\mW=\Psi^{-1} (\body)$ fibered by the particle World lines $\Psi^{-1}(\bX)$.}
  \label{fig:fiberedtimelines}
\end{figure}

\begin{rem}
  It is worth mentioning that the point of view is here reverse to the one of Classical Continuum Mechanics of solids, in which a configuration is an embedding $\pp\colon\body \to \espace$ of the body $\body$ into the 3D space $\espace$. A key difference is that, in Classical Continuum Mechanics, $\pp$ and its tangent map $ \bF=T\pp \colon T\body \to T\espace$, are invertible, whereas here, the matter field $\Psi$ and its tangent map $T\Psi$ are not.
\end{rem}

The pullback by $\Psi$ of the mass measure $\mu$ on the body $\body$
\begin{equation*}
  \Psi^{*} \mu=(\mu\circ \Psi)(T\Psi \cdot, T\Psi \cdot, T\Psi \cdot)
\end{equation*}
is a 3-form defined on the 4D World tube $\mW$. Since $\Psi$ is assumed to be a submersion, there exists a nowhere vanishing vector field $\Jmat$ on $\mW$, such that
\begin{equation}\label{eq:def-P}
  \Psi^{*} \mu=i_{\Jmat} \vol_{g}=\Jmat\cdot \vol_{g},
\end{equation}
where $i_{\Jmat}$ is the interior product (or left contraction) of $\Jmat$ with the 4-form $\vol_{g}$. This vector field $\Jmat$ is the \emph{current of matter}~\cite{Sou1958} (it spans the one-dimensional kernel of the tangent map $T\Psi$).

To describe perfect matter, it is furthermore assumed that $\Jmat$ is timelike, \emph{i.e.} that
\begin{equation*}
  g(\Jmat, \Jmat) <0
  \quad
  \text{on}\; \,\mW,
\end{equation*}
and induces therefore a time orientation on $\mW$. The \emph{rest mass density} is defined as
\begin{equation}\label{eq:def-rho}
  \rho_{r} := \sqrt{-g(\Jmat, \Jmat)}  ,
  \qquad
  \Jmat = \rho_{r} \UU,
\end{equation}
where $\UU$ is a unit timelike vector field on the World tube ($g(\UU,\UU)=-1$). Since $(\dive^{g} \Jmat)\vol_{g} = \rd\Psi^{*} \mu = \Psi^{*} \rd \mu=0$, mass conservation takes the form
\begin{equation*}
  \dive^{g} \Jmat = \dive^{g} (\rho_{r} \UU) = 0
  \quad
  \text{on}\; \,\mW.
\end{equation*}

\subsection{Relativistic variational calculus}

Souriau's formulation of Relativistic Hyperelasticity~\cite{Sou1958,Sou1964} is described by the following Lagrangian,
\begin{equation*}
  \mL[g,\Psi] = \mH[g] + \mL^{\text{mat}}[g,\Psi],
\end{equation*}
where
\begin{equation*}
  \mH[g] = \int \frac{1}{2\kappa}\left(R_{g}(m)-2\Lambda\right)\, \vol_{g}.
\end{equation*}
is the \emph{Hilbert-Einstein functional} which described inertia and gravitation, $R_{g}$ is the scalar curvature, $\Lambda$ is the cosmological constant, and $\kappa=8\pi G/c^{4}$ is the Einstein constant, function of the gravitational constant $G$ and the speed of light $c$. The additional term $\mL^{\text{mat}}[g,\Psi]$, represents perfect matter and is the source of the gravitational field.

A Lagrangian $\mL [g,\Psi]$ is \emph{general covariant} if
\begin{equation*}
  \mL[\varphi^{*}g,\varphi^{*}\Psi] = \mL [g,\Psi],
\end{equation*}
for all (local) diffeomorphism $\varphi$. When a Lagrangian $\mL^{\text{mat}}[g,\Psi]$ is local and its density depends only on the $0$ jet of $g$ and the first jet of $\Psi$, it is written as
\begin{equation}\label{eq:Lmat}
  \mL^{\text{mat}}[g,\Psi]  = \int L (m, g_{m}, \Psi(m), T_{m}\Psi)\, \vol_{g},
\end{equation}
and general covariance implies that
\begin{equation}\label{eq:general-covariance}
  L(m,\bA(m)^{\star}g_{\varphi(m)}\bA(m),\Psi(\varphi(m)),T_{\varphi(m)}\Psi.\bA(m)) = L(\varphi(m),g_{\varphi(m)},\Psi(\varphi(m)),T_{\varphi(m)}\Psi),
\end{equation}
for all $m \in \mM$, where $\bA(m) = T_{m}\varphi$ is the tangent map of the local diffeomorphism $\varphi$. In particular, one can deduce that $L$ does not depend explicitly on $m$ in that case. It was moreover proven by Souriau \cite{Sou1958,Sou1964}, the following result.

\begin{thm}[Souriau, 1958]\label{thm:Souriau}
  Suppose that the Lagrangian
  \begin{equation*}
    \mL^{\text{mat}}[g,\Psi] = \int L(g_{m}, \Psi(m), T_{m}\Psi) \vol_{g},
  \end{equation*}
  is general covariant. Then, its Lagrangian density $L$ can be recast as
  \begin{equation}\label{eq:Lmatgc}
    L(g_{m}, \Psi(m), T_{m}\Psi) = L^{gc}(\Psi(m), \bK(m)),
  \end{equation}
  for some function $L^{gc}$, and where
  \begin{equation}\label{eq:conformation}
    \bK:= (T\Psi)\,g^{-1} (T\Psi)^{\star}
  \end{equation}
  has been called the \emph{conformation} by Souriau.
\end{thm}

The conformation is a vector valued function $\bK: \mM \to \Sym^{2}V$, where $\Sym^{2}V$ is the vector space of symmetric contravariant second-order tensors on $V$. Since $\ker T_{m}\Psi$ is timelike everywhere, $\bK(m)$ is positive definite for all $m \in \mW$. We will see in \autoref{sec:static-spacetime} how the conformation $\bK$ is related to the inverse of the right Cauchy-Green tensor of Classical Continuum Mechanics.

\begin{rem}\label{rem:cometric}
  Since the conformation $\bK$ is naturally a function of the cometric $g^{-1}$, rather than $g$, it may be useful to write the Lagrangian density as $L(g^{-1}, \Psi, T\Psi)$ rather than $L(g, \Psi, T\Psi)$.
\end{rem}

\subsection{Variations relative to the matter field and the Noether stress--energy tensor}
\label{subsec:Psi-variations}

The first variation $\delta_{\Psi} \mL=0$ with respect to the matter field, the vector-valued function $\Psi$, of the Lagrangian
\begin{equation*}
  \mL[g,\Psi] = \mH[g]+ \int L(g_{m}, \Psi(m), T_{m}\Psi)\, \vol_{g},
\end{equation*}
leads to the Euler--Lagrange equation
\begin{equation}\label{eq:EL-Psi}
  \textrm{EL}_{\Psi} := \pd{L}{\Psi}- \dive^{g} \left( \pd{L}{T\Psi}\right)=0, \qquad (\textrm{EL}_{\Psi})_{I} := \left(\pd{L}{\Psi}\right)_{I}- {{\left(\pd{L}{(T\Psi)}\right)_{I}}^{\mu}}_{; \, \mu} = 0.
\end{equation}

\begin{rem}\label{rem:QM}
  In Gauge Theory~\cite{Ble1981}, in which the matter field $\psi$ is the particle wave function of Quantum  Mechanics, the variational equation $\delta_{\psi} \mL=0$ corresponds to the \emph{wave equation} such as the Klein–Gordon equation (the relativistic version of the Schrödinger equation for spinless particles), or the Dirac equation for spin 1/2 particles.
\end{rem}

Assuming that $\Psi$ is a submersion, which is one of the main hypothesis in the theory of perfect matter, the Euler--Lagrange equation \eqref{eq:EL-Psi} is equivalent to the equation
\begin{equation*}
  \dive^{g}  \bTN{\Psi} =  0,
\end{equation*}
by theorem~\ref{thm:Noether-SE-tensor-conservation}, where
\begin{equation*}
  \bTN{\Psi} := (T\Psi)^{\star}\cdot \pd{L}{T\Psi}- L \, \Idq^{\star}, \qquad {\TNcomp{\Psi}_{\mu}}^{\nu}={(T\Psi)^{I}}_{\mu} {\left(\pd{L}{T\Psi}\right)_{I}}^{\nu} - L \, {\delta_{\mu}}^{\nu},
\end{equation*}
is the Noether stress-energy tensor.

\subsection{Variations relative to the metric and the Hilbert stress--energy tensor}

Besides the Noether stress--energy tensor, another stress--energy tensor was introduced by Hilbert in the context of General Relativity \cite{Hil1915}. It is formulated using the variational derivative with respect to the metric $g$ and defined as
\begin{equation}\label{eq:Hilbert-SE}
  \bTH:= -2\frac{\delta \mL^{\text{mat}}}{\delta g},
\end{equation}
when the variational derivative $\delta \mL/\delta g$ is a regular distribution. This tensor $\bTH$, that we shall call the \emph{Hilbert-stress--energy tensor}, is a contravariant second order tensor field, contrary to the Noether stress--energy tensor $\bTN{\Psi}$, defined in \autoref{subsec:Psi-variations} as a mixed tensor field. It is the source term in the Einstein equation
\begin{equation*}
  2\kappa\, \frac{\delta \mH}{\delta g}=\Ein_{g}^{\sharp}+ \Lambda\, g^{-1} =\kappa\, \bTH,
\end{equation*}
obtained as the variational stationary equation $\delta_{g} \mL=0$ for the full Lagrangian $\mL[g,\Psi]$. Here
\begin{equation*}
  \Ein_{g}^{\sharp} = g^{-1}\left(\Ric_{g}-\frac{1}{2}R_{g}\, g\right) g^{-1}
\end{equation*}
is the contravariant form of the Einstein tensor and $\Ric_{g}$ the Ricci tensor. A direct consequence of the general covariance of the Hilbert-Einstein functional $\mH$ is that~\cite{Noe1918}
\begin{equation*}
  \dive^{g} (\Ein_{g}^{\sharp}+ \Lambda\, g^{-1}) = 0,
\end{equation*}
where $\dive^{g}$ is the divergence with respect to the Universe metric $g$. It implies the conservation law~\cite{Ein1915,Wey1917}
\begin{equation}\label{eq:divT}
  \dive^{g} \bTH=0,
\end{equation}
for every solution of the Einstein equation, and is a relativistic extension of the Classical Continuum Mechanics equilibrium equation~\eqref{eq:MMC-Euler--Lagrange-bis}.

\begin{rem}\label{rem:deltaPsi}
  The impressive observation is that for a general covariant Lagrangian $\mL^{\text{mat}}[g,\Psi]$, with density $L^{gc}(\Psi(m), \bK(m))$, we get
  \begin{equation*}
    \bTN{\Psi} = g \bTH,
  \end{equation*}
  Therefore, in the case of general covariant Lagrangian describing perfect matter (see \autoref{sec:variational-calculus}), we have
  \begin{equation}\label{eq:equivA}
    \delta_{\Psi} \mL=0 \iff \textrm{EL}_{\Psi}=0 \iff \dive^{g}  \bTN{\Psi} = 0 \iff \dive^{g} \bTH =0
    \iff \delta_{g} \mL=0,
  \end{equation}
  since the metric is parallel, and $\delta_{\Psi} \mL=0$ does not lead to a new equation.
\end{rem}

\section{Introduction of an observer}
\label{sec:static-spacetime}

In General Relativity, there is no intrinsic definition of \emph{space} and \emph{time}. In order to connect the relativistic framework with the classical one (a Galilean structure on $\mM$~\cite{Kuen1972,Kuen1976}), we need first to introduce a \emph{time function} $\hat{t}$ on $\mM$. This time function is assumed to have a timelike gradient everywhere (and is thus a submersion). It is usually interpreted as the introduction of an observer. When such a function exists, it allows to define a \emph{spacetime structure} or \emph{$(3+1)$-structure} on $\mM$ (see~\cite{Dar1927,Lic1939,FBCB1948,Lic1952,Fou1956,ADM1962,Yor1979,Gou2012}). This structure corresponds to a foliation of $\mM$ by spacelike hypersurfaces
\begin{equation*}
  \espace_{t} := \set{m \in \mM;\; \hat{t}(m) = t}.
\end{equation*}
The unit timelike normal $\NN$ to $\espace_{t}$ is defined as
\begin{equation}\label{eq:defN}
  \NN := -\mathcal{N} \gradg \hat{t},
  \qquad
  \mathcal{N} := \left(- g\big(\gradg \hat{t},\gradg \hat{t}\big) \right)^{-1/2},
\end{equation}
where $\mathcal{N}$ is called the \emph{lapse function}~\cite{ADM1962,Gou2012} and the gradient is relative to the metric $g$. The minus sign is chosen so that the quadrivector $\NN$ is \emph{future-oriented}, meaning that the value of $\hat{t}$ increases along the flow curves of $\NN$.

We shall introduce the 3D submanifolds of the World tube $\omega_{t}:= \mW \cap \espace_{t}$ and denote by $\Psi_{t}$,  the restriction of $\Psi$ to $\omega_{t}$. Here, the submanifolds $\omega_{t}$, parameterized by time $t$ play the same role as the configuration systems in Classical Continuum Mechanics, see \autoref{fig:foliationOmegat}. Since we have assumed that $\Psi$ is a submersion and that $\ker T\Psi$ is timelike, the tangent map $T\Psi_{t} \colon T \omega_{t} \to T\body$ is invertible. In practice, there is no obstruction to assume that $\Psi_{t} \colon \omega_{t} \to \body$ is moreover a diffeomorphism.

\begin{figure}[ht]
  \centering
  \includegraphics[width=12cm]{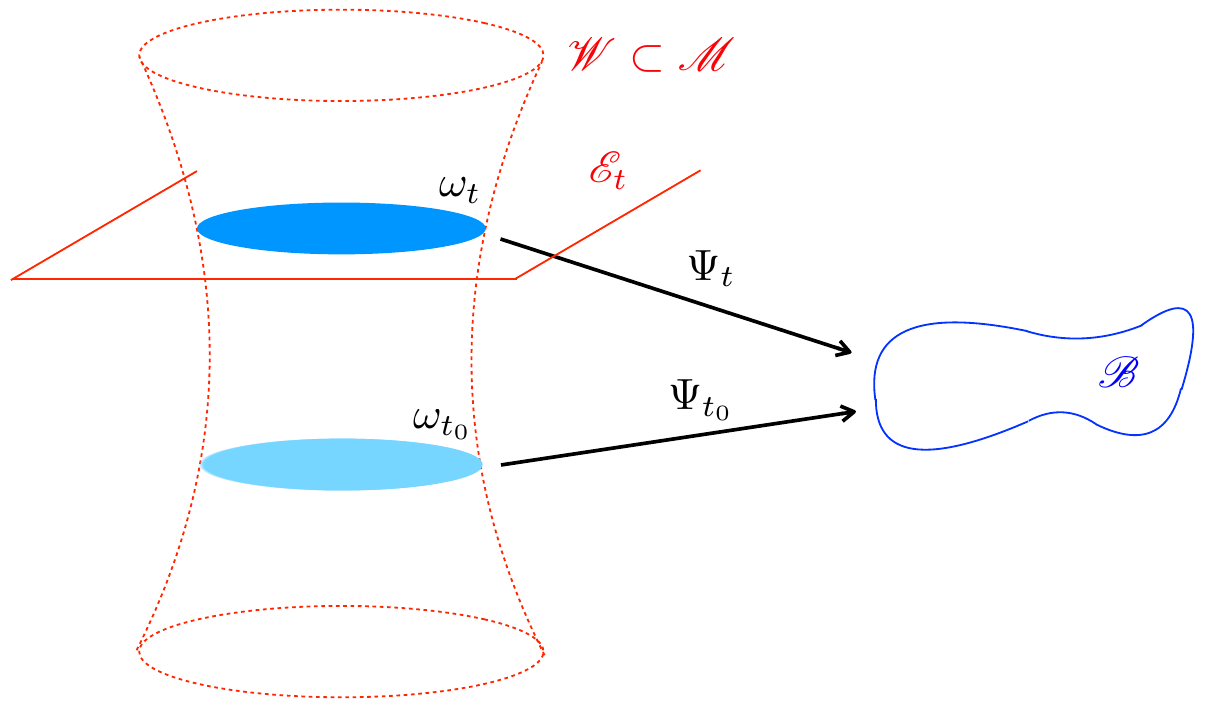}
  \caption{The foliation of the World tube $\mW$ by spacelike hypersurfaces $\omega_{t}$.}
  \label{fig:foliationOmegat}
\end{figure}

The orthogonal decomposition with respect to $\NN$ of the unit timelike quadrivector $\UU$ (defined by~\eqref{eq:def-rho}) is written as
\begin{equation}\label{eq:UN}
  \UU = \gamma \left(\NN+ \frac{\uu}{c}\right), \qquad g(\NN, \uu) = 0.
\end{equation}
It allows to define the (generalized) \emph{Lorentz factor} $\gamma$ as the scalar product~\cite{Gou2012}
\begin{equation}\label{eq:gamma}
  \gamma: = - g(\UU, \NN) = \frac{1}{\sqrt{1-\frac{\norm{\uu}_{g}^{2}}{c^{2}}}} ,
\end{equation}
and to interpret $\uu$ as the relativistic 3D Eulerian velocity~\cite{Sou1964,KD2023}. Furthermore, from $T\Psi. \UU=0$, we get that
\begin{equation*}
  T\Psi.\left(\NN+ \frac{\uu}{c}\right) = 0, \quad \text{with} \quad T\Psi.\NN = -\mathcal{N} \partial_{t} \Psi,
\end{equation*}
and we obtain the following expression for the relativistic Eulerian velocity,
\begin{equation}\label{eq:uRG}
  \uu = -\frac{1}{\mathcal{N}}\, T\Psi_{t}^{-1}.\partial_{t} \Psi.
\end{equation}

\subsection{Static spacetimes}
\label{subsec:static-spacetimes}

In order to make closer the connection with Classical Continuum Mechanics, we choose first to adopt the \emph{passive matter} hypothesis, meaning that the presence of a continuous medium/structure does not affect the Universe metric $g$, which can then be assumed as given solution of the Einstein equations in the vacuum (a so-called exterior solution). We shall assume moreover that this metric $g$ is \emph{static}~\cite{Lic1955,HE1973,MTW1973,MG2009}, which means that
\begin{enumerate}
  \item there exists a Killing vector field $\MM$ on the Universe (\emph{i.e.} $\Lie_{\MM}g = 0$) which is timelike everywhere (\emph{stationary metric} \textit{à la} Lichnerowicz~\cite{Lic1955}),
  \item and moreover that $\MM = -\gradg \hat{t}=-g^{-1} \rd \hat{t}$ is the gradient (with respect to the metric $g$) of a \emph{time} function~$\hat{t}$, the minus sign making it future oriented.
\end{enumerate}

In that case, the flow of the vector field $\MM$ acts as a one-parameter sub-group $G$ of isometries for the metric $g$ (which could be interpreted as time translations). We will make the following further global hypotheses which seem reasonable for our purpose (some of the following assertions may be redundant, see~\cite{Har2015} and bibliography herein for more details on stationary and static spacetimes):
\begin{enumerate}
  \item The Killing vector field $\MM$ is \emph{complete}, meaning that integral curves are all defined on $(-\infty,+\infty)$;
  \item The one-parameter subgroup $G$ generated by the flow of $\MM$ acts freely on $\mM$;
  \item Each integral curve of $\MM$ is a line (not a circle) and intersects each hypersurface $\espace_{t}$ in one and only one point;
  \item $G$ permutes the 3D hypersurfaces~$\espace_{t}$.
\end{enumerate}

The quotient space $\espace := \mM/G$ is a smooth three-dimensional manifold, which can be interpreted as the 3D space of Classical Mechanics. Moreover, under our global hypothesis, the following map
\begin{equation}\label{eq:universe-trivialization}
  \btau \colon \mM \to \RR \times \espace, \qquad m \to (\hat{t}(m), \pi(m))
\end{equation}
where $\pi \colon \mM \to \espace$ is the canonical projection map, is a diffeomorphism which allows to trivialize the universe as a global spacetime \emph{i.e.}, (3+1)-structure.

The Universe metric $g$ induces then a ---not necessarily flat--- Riemannian metric $\bar{g}$ on the quotient space $\espace$ such that the canonical projection
\begin{equation}\label{eq:pi}
  \pi \colon (\mM,g) \to (\espace, \bar{g})
\end{equation}
is a Riemannian submersion~\cite{GHL2004}. In particular, its restriction to each 3D submanifold $\pi \colon \espace_{t} \to \espace$ is a Riemannian isometry.

Given a static metric $g$, it is always possible to find a local coordinate system $(x^{\mu})$ such that
\begin{equation}\label{eq:staticxmu}
  \pd{g_{\mu\nu}}{x^{0}} = 0, \quad \text{and} \quad g_{0i} = 0 \quad (1 \le i \le 3).
\end{equation}
In practice, such a coordinate system is introduced \textit{a priori}~\cite{MG2009} and is chosen such that the time function is $\hat{t}= x^{0}/c$, where $c$ is the speed of light. The spacelike hypersurfaces $\espace_{t}$ correspond thus to the level sets $x^{0}=ct$. By \eqref{eq:defN}, the lapse function is then $\mathcal{N} = \sqrt{-g_{00}}$ and
\begin{equation*}
  \frac{\partial}{\partial x^{0}} = \frac{1}{c} \frac{\partial}{\partial t}=\mathcal N\, \NN.
\end{equation*}
so that the metric can be expressed as
\begin{equation*}
  g = -\NN^{\flat} \otimes \NN^{\flat} + \mg = -\mathcal{N}^{2} c^{2} \rd t^{2} + \mg,
  \qquad
  \mg \NN=0,
\end{equation*}
where $\mg= g_{ij} \rd x^{i} \rd x^{j}$ is the spatial part of $g$. Finally, the 4D volume form associated to the Universe metric $g$ can  be expressed as
\begin{equation*}
  \vol_{g} = - \NN^{\flat} \wedge  \vol_{\mg}= \mathcal{N} c\, \rd t \wedge \vol_{\mg},
  \quad \text{where} \quad
  \vol_{\mg} := i_{\NN} \vol_{g}.
\end{equation*}
Note that $\vol_{\mg}$ is a time independent volume form on each spatial hypersurface $\espace_{t}$.

Examples of static spacetimes (satisfying our global hypotheses) are given by the flat \emph{Minkowski spacetime} and the (exterior) \emph{Schwarzschild spacetime} (to take account of gravity, as in~\cite{KD2023}).

\begin{exam}[The Schwarzschild spacetime]\label{ex:Schwarzschild}
  It was formulated in~\cite{Sch1916} to describe the gravitational field around a spherical celestial object such as the sun and used to prove the first experimental confirmations of General relativity. The Schwarzschild metric can be expressed as
  \begin{equation*}
    g = -\mathcal{N}(r) \, c^{2} \rd t^{2}+k(r) \,\bq,
    \qquad
    \begin{cases}
      \bq = \delta_{ij} \rd x^{i} \rd x^{j},
      \\
      r=\sqrt{\delta_{ij} x^{i} x^{j}},
    \end{cases}
  \end{equation*}
  where $\mg =k(r) \, \bq$ is conformal to the Euclidean metric in the Cartesian isotropic coordinates $(x^{0}=ct, x^{i})$. The associated 4D volume form on $\mM$ is
  \begin{equation*}
    \vol_{g}= \mathcal{N} k^{\frac{3}{2}} c \rd t \wedge \vol_{\bq},
  \end{equation*}
  and we refer to \cite{MG2009,KD2023} for the corresponding expressions of the lapse function $\mathcal{N}=\mathcal{N}(r)$ and of the conformal factor $k=k(r)$. Let us simply recall that~\cite{MG2009}
  \begin{equation}\label{eq:Schwarzschild}
    \mathcal{N}= 1 -\frac{1}{c^{2}} \frac{GM}{r_{0}}+O\left(\frac{1}{c^{4}}\right) \approx 1,
    \qquad
    k= 1 +\frac{2}{c^{2}} \frac{GM}{r_{0}}+O\left(\frac{1}{c^{4}}\right) \approx 1,
  \end{equation}
  at the surface of a planet of mass $M$ and of radius $r_{0}$, where $G$ is the gravitational constant and $c$ the speed of light.
\end{exam}

\begin{exam}[The Minkowski spacetime]\label{ex:Minkowski}
  It corresponds to Special Relativity, which is defined as the affine space $\RR^{4}$ endowed with the flat Lorentzian metric
  \begin{equation*}
    g = \eta :=-c^{2} \rd t^{2}+\bq,
  \end{equation*}
  where $\mg = \bq = \delta_{ij} \rd x^{i} \rd x^{j}$ is the Euclidean metric, and the lapse function is $\mathcal{N}=1$. In that case, the Killing field is $\MM = \frac{\partial}{\partial t}$ and the 3D quotient space $\espace=\mM/G$ is endowed with the Riemannian metric $\bar{g}=\bq$. The 4D volume form on $\mM$ is then given by
  \begin{equation*}
    \vol_{g}= \vol_{\eta}= c \rd t \wedge \vol_{\bq}.
  \end{equation*}
  Here $\MM = \frac{\partial}{\partial t}$ and $G$ corresponds really to times translations, and the quotient space $(\espace,\bar{g})$ corresponds to the affine Euclidean space of Classical Continuum Mechanics with $\bar{g}=\bq$. In that case, the trivialization $\btau$ is a Riemannian isometry between $(\mM,\eta)$ and the product manifold
  \begin{equation*}
    (\RR, -c^{2} {\rd t}^{2}) \times (\espace,\bq).
  \end{equation*}
\end{exam}

\subsection{The observer as a matter field}
\label{subsec:observer-matter-field}

Following Souriau~\cite[Section 39]{Sou1964}, when a static spacetime is involved, it is pertinent to give a more precise definition of an ``observer'', not only as a time function (as it seems common in General Relativity), but as a special ---observer--- perfect matter field $\Psi_{0}$, ``linked'' with the spacetime structure induced by the time function $\hat{t}$. More precisely, we introduce such a matter field as a vector valued function $\Psi_{0} \colon \mM \to V$ which satisfies moreover
\begin{equation*}
  \Lie_{\MM}\Psi_{0} = 0,
\end{equation*}
where $\MM = -\gradg \hat{t}$ is the Killing vector defining the static spacetime structure. This ``observer'' (which can be assimilated as the laboratory) is thus synchronized with the time function $\hat{t}$ (its proper time coincides with $t$), since $T\Psi_{0}. \gradg \hat{t} = 0$. The world lines of $\Psi_{0}$ coincide with the flow lines of $\MM$. In that case, $\Psi_{0}$ induces a diffeomorphism
\begin{equation*}
  \overline{\Psi}_{0} \colon \espace  \to V, \quad \text{such that} \quad \Psi_{0} = \overline{\Psi}_{0} \circ \pi,
\end{equation*}
in such a way that the following diagram commutes
\begin{equation*}
  \xymatrix{
  \mM \ar[d]^{\pi} \ar[r]^{\Psi_{0}} & V  \\
  \espace  \ar@/_/[ur]_{\overline{\Psi}_{0}}  }
\end{equation*}

In this framework, we can then propose a relativistic generalization of the Classical Continuum Mechanics reference embedding
\begin{equation*}
  \pp_{0}:={\overline\Psi}_{0}^{\, -1} \colon \body \to \espace,
\end{equation*}
which corresponds to the diagram
\begin{equation*}
  \xymatrix{
  \mM \ar[d]^{\pi} \ar[r]^{\Psi_{0}} &
  V \ar@/^/[dl]^{\pp_{0}} \\
  \espace
  }
\end{equation*}
Then, in accordance with classical 3D notations, we shall introduce
\begin{equation*}
  \Omega_{0} := \pp_{0}(\body)  \subset \espace,
\end{equation*}
interpreted as the reference configuration system embedded in the three-dimensional Riemannian manifold $(\espace, \bar{g})$, and
\begin{equation*}
  \bF_{0} := T\pp_{0},
\end{equation*}
the linear tangent map of $\pp_{0} = {\overline{\Psi}_{0}}^{-1}$.

\begin{rem}
  In practice, it is often assumed that the restriction of the observer matter field $\Psi_{0}$ to the spacelike hypersurface $\espace_{t_{0}}$ coincides with $\Psi_{t_{0}}$ at some ``initial date'' $t=t_{0}$. In that case we get
  \begin{equation*}
    \Omega_{0} = \pi(\omega_{t_{0}}).
  \end{equation*}
\end{rem}

\subsection{Effective matter fields and relativistic deformation}
\label{subsec:effective-matter-field}

The choice of an observer matter field~$\Psi_{0}$, and the induced diffeomorphism $\pp_{0} = {\overline{\Psi}_{0}}^{-1} \colon V \to \espace$, allows to replace the field variable $\Psi \colon \mM \to V$ by the field variable $\widetilde{\Psi} \colon\mM \to \espace$~\cite{Sou1964}, where
\begin{equation}\label{eq:tildePsi}
  \widetilde{\Psi} := \pp_{0} \circ \Psi,
\end{equation}
is called the \emph{effective matter field}. It is defined on the 4D Universe $\mM$ but with values in the 3D manifold $\espace$, instead of the vector space $V$ for $\Psi$.

\begin{rem}
  There is an apparent complexification of the problem since $V$ is assumed to be a vector space but $\espace$ is only a manifold. Note however, that if we can assume that $\overline{\Psi}_{0} \colon \espace \to V$ is a diffeomorphism, then $\espace$ inherits the structure of an affine space. This remark is important since we deal with a new variable $\widetilde{\Psi}$ of a variational calculus.
\end{rem}

The restriction of the effective matter field $\widetilde{\Psi}$ to the open subset $\omega_{t}$ of $\espace_{t}$
\begin{equation*}
  \widetilde{\Psi}_{t} \colon  \omega_{t} \to  \Omega_{0}= \pp_{0}(\body), \qquad \widetilde{\Psi}_{t} := \pp_{0} \circ \Psi_{t},
\end{equation*}
is a diffeomorphism. It allows us to define the \emph{relativistic three-dimensional deformation} $\phi=(\phi_{t}=\phi(t))$, parameterized by time $t$, and given by
\begin{equation}\label{eq:phiRG}
  \phi_{t} \colon \Omega_{0} \to \Omega_{t}:=\pi(\omega_{t}),
  \qquad
  \phi_{t} := \pi \circ {\widetilde{\Psi}_{t}}^{-1},
\end{equation}
The 3D submanifolds of $\espace$,
\begin{equation*}
  \Omega_{t}=\pi(\omega_{t})
  \quad \text{and} \quad
  \Omega_{0}= \pi(\omega_{t_{0}})
\end{equation*}
correspond to the deformed and reference configuration systems of Continuum Mechanics (or $\Omega_{\pp(t)}$ and $\Omega_{0}$ in \autoref{sec:classical-elastodynamics}). Similarly, the mapping $\pp=(\pp_{t}=\pp(t))$, parameterized by time,
\begin{equation}\label{eq:p}
  \pp_{t} := \pi \circ \Psi_{t}^{-1} \colon \body \to \Omega_{t} \subset \espace
\end{equation}
is an embedding, interpreted in Continuum Mechanics as the embedding of the body $\body$ into the Euclidean space $(\espace, \bq)$. The mappings $\pp$ and $\phi$ are linked by the relation
\begin{equation*}
  \phi_{t} = \pp_{t} \circ \pp_{0}^{-1},
\end{equation*}
identical to \eqref{eq:phi3D} for Classical Continuum Mechanics, where $\pp_{0} = {\overline{\Psi}_{0}}^{-1}$ can be interpreted as a reference configuration.

\begin{rem}
  Assuming that the restriction of the observer field $\Psi_{0}$ to the spacelike hypersurface $\omega_{t_{0}}$ coincides at an initial time $t=t_{0}$ with $\Psi_{t_{0}}$, we get the common identities of Classical Continuum Mechanics,
  \begin{equation*}
    \pp_{t_{0}} = \pp_{0}, \qquad  \Omega_{t_{0}} = \Omega_{0}  \quad \text{and} \quad \phi_{t_{0}} = \id,
  \end{equation*}
  but in a relativistic framework.
\end{rem}

The corresponding linear tangent maps are given by
\begin{equation*}
  \bF := T\pp_{t} \colon T\body \to T\Omega_{t}, \qquad \bF_{\phi} := T\phi_{t} \colon  T\Omega_{0} \to T\Omega_{t},
\end{equation*}
and we have
\begin{equation*}
  \bF^{-1} = T\Psi_{t}, \qquad \bF_{\phi} = \bF \bF_{0}^{-1} = (T\widetilde{\Psi}_{t})^{-1},
\end{equation*}
so that $\bF_{\phi}$ appears as a relativistic version of the so-called deformation gradient in Classical Continuum Mechanics.

One can furthermore introduce a \emph{relativistic four-dimensional deformation $\Phi$} as
\begin{equation}\label{eq:defPhi}
  \Phi \colon  \RR \times \Omega_{0} \to \RR \times \espace,
  \qquad
  (t, \xx_{0}) \mapsto  (t, \xx)= \Phi(t, \xx_{0})= (t, \phi_{t}(\xx_{0})),
\end{equation}
which might be seen as a diffeomorphism between open subsets of $\mM$, using the trivialization~\eqref{eq:universe-trivialization}. Since
\begin{equation*}
  \widetilde{\Psi}(t,\phi_{t}(\xx_{0})) = \xx_{0}, \qquad \forall t \in \RR, \, \forall \xx_{0} \in \Omega_{0},
\end{equation*}
by the very definition of $\phi_{t}$ and using the trivialization~\eqref{eq:universe-trivialization}, we get that
\begin{equation*}
  \Phi^{-1}(t, \xx)= (t,  \widetilde{\Psi}_{t}(\xx)),
\end{equation*}
and, in time-space block notation,
\begin{equation}\label{eq:TPhiblock}
  T\Phi = \begin{bmatrix}
    1                & 0          \\
    \partial_{t}\phi & \bF_{\phi}
  \end{bmatrix}
  = \begin{bmatrix}
    1                                                    & 0          \\
    -\bF_{\phi}.(\partial_{t}\widetilde{\Psi})\circ \Phi & \bF_{\phi}
  \end{bmatrix}
\end{equation}
since
\begin{equation*}
  \partial_{t}\widetilde{\Psi}(t,\phi_{t}(\xx_{0})) + (T\widetilde{\Psi}_{t})\partial_{t}\phi = 0 \quad \text{and} \quad T\widetilde{\Psi}_{t} . \bF_{\phi} = \bI_{3}.
\end{equation*}

By~\eqref{eq:uRG}, we get the following equivalent expressions of the relativistic Eulerian velocity
\begin{equation*}
  \uu(t) = -\frac{1}{\mathcal{N}} T\Psi_{t}^{-1}.\partial_{t} \Psi = -\frac{1}{\mathcal{N}} T\widetilde{\Psi}_{t}^{-1}.\partial_{t} \widetilde{\Psi}
  =\frac{1}{\mathcal{N}} \partial_{t} \phi \circ \Phi^{-1},
\end{equation*}
either in terms of the matter field $\Psi$, the effective matter field $\widetilde{\Psi}$, or the deformation $\phi$.

Introducing the (relativistic) Lagrangian velocity $\VV:=\partial_{t} \phi$, we shall summarize the notions of relativistic Eulerian and of Lagrangian velocities with the following diagram
\begin{equation*}
  \xymatrix{
  T \Omega_{0} \ar[r]^{\bF_{\phi}} \ar[d]  & {T \Omega_{t}}\ar[d] &
  \\
  \Omega_{0} \ar[ur]^{\qquad\VV} \ar[r]^{\phi_{t}}    &  \Omega_{t}  \ar@/_1pc/[u]_{\mathcal{N} \uu(t)} & \hskip -15mm
  }
\end{equation*}
which shows that the deformation $\phi_{t}$, its tangent map $\bF_{\phi}=T\phi_{t}$, the Lagrangian and Eulerian velocities $\bV=\partial_{t}\phi$ and $\mathcal{N} \uu$ are the relativistic generalization of the same classical notions. Recall that $\mathcal{N}=1$ for Minkowski spacetime and that $\mathcal{N} \approx 1$ (by Schwarzschild lapse function~\eqref{eq:Schwarzschild}) at the surface of a planet.

\subsection{Effective conformation}
\label{subsec:effective-conformation}

To make a closer link with Classical Continuum Mechanics, Souriau has further introduced the \emph{effective conformation}~\cite[\S 39]{Sou1964}
\begin{equation}\label{eq:KtildeA}
  \widetilde{\bK} := (T\widetilde{\Psi})\, g^{-1} (T\widetilde{\Psi})^{\star},
  \qquad
  \mM \to \Sym^{2}T^{\star}\espace.
\end{equation}
We have then
\begin{equation*}
  \widetilde{\bK}= T\widetilde{\Psi}_{t}\left(\mg^{\sharp} - \frac{1}{c^{2}} \uu \otimes \uu \right)(T\widetilde\Psi_{t})^{\star},
\end{equation*}
and
\begin{equation}\label{eq:KtildeB}
  \widetilde{\bK} \circ \Phi
  =  \bF_{\phi}^{-1} \left(\bar{g}^{-1}-  \frac{1}{\mathcal{N}^{2} c^{2}} \VV \otimes \VV \right)\bF_{\phi}^{-\star}
  =  (\bF_{\phi}^{\star}\,\bar{g}\,\bF_{\phi})^{-1}-  O\left(\frac{1}{c^{2}}\right),
\end{equation}
which is defined on $\RR \times \Omega_{0}$.

\begin{rem}
  Assuming that the observer matter field $\Psi_{0}$ and the matter field $\Psi$ coincide at an initial date $t=t_{0}$ on the spacelike manifold $\omega_{t_{0}}$, we have  $\phi(t=t_{0})=\id$ and $\bF_{\phi}(t=t_{0}) = \bI_{3}$,
  so that
  \begin{equation*}
    \widetilde{\bK}(t=t_{0}) = \bar{g}^{-1}+ O\left(\frac{1}{c^{2}}\right).
  \end{equation*}
\end{rem}

This last expression allows us to recognize $\bF_{\phi}^{\star}\,\bar{g}\,\bF_{\phi}$ as the right Cauchy--Green tensor, when $\mathcal{N} = k = 1$ and $\bar{g}=\bq$ is the Euclidean metric. This is the case in Special Relativity. For this reason, the second order time dependent covariant tensor field
\begin{equation}\label{eq:genC}
  \bC:=\widetilde{\bK}^{-1}\!\!\circ \Phi,
\end{equation}
defined on $\RR \times \Omega_{0}$ can be considered as the relativistic generalization of the right Cauchy--Green tensor for more general static spacetime (satisfying our global hypothesis of \autoref{subsec:static-spacetimes}).

\subsection{Mass conservation for static spacetimes}
\label{subsec:mass-conservation}

In a static spacetime, the metric $g$ can be written as
\begin{equation*}
  g = -\mathcal{N}^{2} c^{2} \rd t^{2}+\mg,
\end{equation*}
where $\mathcal{N}$ is the laps function, and $\mg$ is the 3D Riemannian metric on the spatial hypersurfaces. The 4D volume form on $\mM$ is written thus as
\begin{equation*}
  \vol_{g} = \mathcal{N} c\, \rd t \wedge \vol_{\mg}.
\end{equation*}
Since $\pi \colon (\mM,g) \to (\espace, \bar{g})$ is a Riemannian submersion and $\ker T \pi$ is generated by $\grad \hat{t}$, we have that $\pi^{\star} \bar{g} =\mg$ and $\pi^{\star} \vol_{\bar{g}} = \vol_{\mg}$, where $\bar{g}$ is the induced Riemannian metric on $\espace$.

The quadrivelocity $\UU$ defined in~\eqref{eq:def-rho}, and the Lorentz factor $\gamma$, defined in~\eqref{eq:gamma}, are given respectively by
\begin{equation*}
  \UU = \frac{\gamma}{c}\left(\frac{1}{\mathcal{N}} \pd{}{t} + \uu\right),
  \quad \text{and} \quad
  \gamma = \frac{1}{\sqrt{1-\frac{\norm{\uu}_{\mg}^{2}}{c^{2}}}}.
\end{equation*}
The \emph{relativistic mass density} is defined as \cite{KD2023}
\begin{equation}\label{eq:relativistic-mass-density}
  \rho:= \rho_{r} \gamma,
\end{equation}
where $\rho_{r}$ was defined as the rest mass density and related to the current of matter~\eqref{eq:def-P} by~\eqref{eq:def-rho}. Introducing a reference embedding $\pp_{0}=\overline \Psi_{0}^{-1}\colon \body \to \espace$, we get
\begin{equation*}
  (\pp_{0})_{*}\mu = \rho_{0} \vol_{\bar{g}},
\end{equation*}
where the \emph{reference mass density} $\rho_{0}=\rho_{0}(\xx_{0})$ is defined on the reference configuration $\Omega_{0}\subset \espace$.  We have then the following result.

\begin{lem}[Mass conservation]\label{lem:mass-conservation}
  Let $g$ be a static metric and $\Phi$, the relativistic transformation introduced in~\eqref{eq:defPhi}. Then, a relativistic mass conservation law can be expressed as
  \begin{equation*}
    \left(\rho \circ \Phi\right) J_{\Phi}^{g} = \rho_{0} \circ \pi,
  \end{equation*}
  where $\pi \colon \mM \to \espace$ is the projection map and the Jacobian $J_{\Phi}^{g}$ is defined intrinsically by the relation $\Phi^{*}\vol_{g} = J_{\Phi}^{g} \, \vol_{g}$.
\end{lem}

\begin{rem}
  The present formulation of mass conservation on the reference configuration system $\Omega_{0}\subset \espace$ generalizes the expression \eqref{eq:MC}, $\rho_{0}= (\rho \circ \phi) J_{\phi} = (\rho \circ \phi) \det \bF_{\phi}$, for Classical Continuum Mechanics, for which $\espace$ is the affine Euclidean space.
\end{rem}

\begin{proof}
  We have
  \begin{equation*}
    \Phi_{*}\left(\pd{}{t}\right) = T \Phi.\pd{}{t} \circ \Phi^{-1} = \pd{}{t} + \mathcal{N} \uu,
  \end{equation*}
  and since
  \begin{equation*}
    \Jmat = \rho_{r} \UU = \frac{\gamma}{c}\rho_{r}\left(\frac{1}{\mathcal{N}}\pd{}{t} + \uu\right) = \frac{\rho}{\mathcal{N}c}\left(\pd{}{t} + \mathcal{N}\uu\right),
  \end{equation*}
  we get thus
  \begin{equation*}
    \Phi^{*}\Jmat = \left(\frac{\rho}{\mathcal{N} c} \circ \Phi\right)  \pd{}{t}.
  \end{equation*}
  We have therefore
  \begin{equation*}
    \Phi^{*}(i_{\Jmat} \vol_{g}) = i_{\Phi^{*}\Jmat} \, \Phi^{*}\vol_{g} = \left(\frac{\rho}{\mathcal N c} \circ \Phi\right) J_{\Phi}^{g} \, i_{\pd{}{t}} \vol_{g} = (\rho \circ \Phi)J_{\Phi}^{g} \vol_{\mg},
  \end{equation*}
  on one hand, whereas, on the other hand,
  \begin{equation*}
    \Phi^{*}(i_{\Jmat} \vol_{g}) = \Phi^{*} \Psi^{*}\mu = \Phi^{*}\widetilde{\Psi}^{*} (\rho_{0} \vol_{\bar{g}}) = \pi^{*}(\rho_{0} \vol_{\bar{g}}) = (\rho_{0} \circ \pi) \vol_{\mg}.
  \end{equation*}
  We get thus $(\rho \circ \Phi) J_{\Phi}^{g} = \rho_{0} \circ \pi$, which achieves the proof.
\end{proof}

\section{From variations with respect to the matter field to configurational forces balance}
\label{sec:relativistic-configurational-forces}

The Lagrangian $\mL[g,\Psi] = \mH[g] + \mL^{\text{mat}}[g,\Psi]$, introduced in \autoref{sec:relativistic-elastodynamics} depends on two fundamental variables, the metric $g$ of the universe and the matter field $\Psi$, which represents perfect matter. The first variation with respect to the metric $g$ leads to Einstein's equation, and to the conservation law $\dive^{g} \bTH=0$, for Hilbert stress-energy tensor. It has been shown in~\cite{Sou1958,Sou1964} for Minkowski's spacetime and in~\cite{KD2023} for Schwarzschild's spacetime, that as $c \to \infty$, one recovers from it, Newton's momentum balance of Classical Continuum Mechanics, the latter with gravity. The first variation of the Lagrangian with respect to the matter field $\Psi$ seems to have not received any interpretation in this context. In Gauge Theory, when applied to Quantum Mechanics (Remark~\ref{rem:QM}), it is known that the first variation with respect to the elementary particles matter field leads to the wave equation~\cite{Ble1981}, but what is the interpretation of the equation $\delta_{\Psi} \mL=0$ in present relativistic elastodynamics?

In this section, we shall answer this question, by showing that the variational equation $\delta_{\Psi} \mL=0$ leads to the usual configurational forces balance, at the classical limit. An intermediate step to establish this result will be the formulation of relativistic configurational forces balance in Minkowski spacetime, that is in Special Relativity.

To achieve this goal, the first step is to introduce an observer matter field $\Psi_{0}$, as detailed in \autoref{subsec:observer-matter-field}, and compatible with a static spacetime structure defined by a Killing vector $\MM = -\gradg \hat{t}$, where $\hat{t}$ is the time function. This allows us to define the space $\espace$ as the quotient space of the universe $\mM$ by the flow of $\MM$ and to introduce the projection map $\pi \colon \mM \to \espace$ onto this space. Under certain hypotheses (see \autoref{subsec:static-spacetimes}), it is possible to trivialize the universe as the global spacetime $\RR \times \espace$, and to regard $\espace$ as the affine space of Continuum Mechanics; a point $m \in \mM$ corresponds then to a couple $(t,\xx)$, where $\xx \in \espace$.

Since $\Psi_{0}$ is invariant under the flow of $\MM$ (see \autoref{subsec:observer-matter-field}), it induces a diffeomorphism $\overline{\Psi}_{0} \colon \espace \to V$, where $V$ is the 3D space which labels particles. Its inverse
\begin{equation*}
  \pp_{0} := {\overline{\Psi}_{0}}^{-1} \colon \body \subset V \to \Omega_{0} \subset \espace
\end{equation*}
is the relativistic analog of a reference configuration in Classical Continuum Mechanics (\autoref{sec:classical-elastodynamics}). It allows us to redefine the matter field $\Psi\colon \mM \to V$ as a map $\widetilde{\Psi}\colon \mM \to \espace$, the \emph{effective matter field} (see \autoref{subsec:effective-matter-field}), where $\widetilde{\Psi} := \pp_{0} \circ \Psi$, and $\pp_{0}$ is considered as a parameter.

\begin{rem}
  From the point of view of variational calculus, the substitution $\Psi \mapsto \widetilde{\Psi}$ is just a change of variable. However, $V$ was assumed to be a vector space and $\Psi$, a vector valued function, but $\espace$ is, \textit{a priori}, just a manifold. However, the additional assumption that $\pp_{0} \colon V \to \espace$ is a diffeomorphism, allows us to pushforward the vectorial structure of $V$ onto $\espace$, eliminating this difficulty.
  Hence, $\widetilde{\Psi}$ can also be considered as a vector valued function. This is fortunate because variational calculus is more involving for maps which are not vector valued function or sections of a vector bundle.
\end{rem}

We can thus formulate two variational problems
\begin{enumerate}
  \item $\delta_{\Psi} \mL = 0$, where $\mL[g, \Psi] = \mH[g] + \int L(g, \Psi,T\Psi)\vol_{g}$, with Noether stress--energy tensor
        \begin{equation*}
          \bTN{\Psi} := (T\Psi)^{\star} \cdot \pd{L}{T\Psi} - L\, \Idq^{\star},
        \end{equation*}
  \item and, $\delta_{\widetilde{\Psi}} \widetilde{\mL} = 0$, where $\widetilde{\mL} = \mH[g] + \widetilde{\mL}[g, \widetilde{\Psi}; \pp_{0}] =
          \int \widetilde L(g, \widetilde\Psi,T\widetilde\Psi)\vol_{g}$, with Noether stress--energy tensor
        \begin{equation*}
          \bTN{\widetilde{\Psi}} := (T\widetilde{\Psi})^{\star} \cdot \pd{\widetilde{L}}{T\widetilde{\Psi}} - \widetilde{L}\, \Idq^{\star}.
        \end{equation*}
\end{enumerate}

Notice that
\begin{equation*}
  \widetilde{L}(g(m),\widetilde{\Psi}(m),T_{m}\widetilde{\Psi}) = L(g(m),\Psi(m),T_{m}\Psi), \qquad \forall m \in \mM,
\end{equation*}
or more simply
\begin{equation}\label{eq:Ltilde}
  \widetilde{L}(g,\widetilde{\Psi},T\widetilde{\Psi}) = L(g,\Psi,T\Psi),
  \qquad \text{on $\mM$},
\end{equation}
with $\widetilde{\Psi} = \pp_{0} \circ \Psi$. We deduce, then, that
\begin{equation*}
  T\widetilde{\Psi} = \bF_{0}. T\Psi, \quad \text{and} \quad \pd{\widetilde{L}}{T\widetilde{\Psi}} = \bF_{0}^{-\star}. \pd{L}{T\Psi},
\end{equation*}
where $\bF_{0}=T\pp_{0}$. Therefore, the two Noether stress--energy tensors are equal,
\begin{equation*}
  \bTN{\widetilde{\Psi}} =  \bTN{\Psi},
\end{equation*}
and
\begin{equation}\label{eq:equivB}
  \dive^{g} \bTN{\widetilde{\Psi}} = 0 \iff  \dive^{g}  \bTN{\Psi} = 0.
\end{equation}
This shows that the two variational problems are equivalent (see \autoref{sec:variational-calculus}).

\subsection{Relativistic Eshelby tensor}

We shall now reformulate the variational problem~(2) using the change of variable given by the four-dimensional deformation~\eqref{eq:defPhi}
\begin{equation*}
  \Phi \colon  \RR \times \Omega_{0} \to \RR \times \espace,
  \qquad
  m_{0} \mapsto m =\Phi(m_{0}).
\end{equation*}
We get
\begin{equation*}
  \widetilde{\mL}[g, \widetilde{\Psi}; \pp_{0}] = \int \widetilde{L}(g,\widetilde{\Psi},T\widetilde{\Psi})(m) \, \vol_{g}
  = \int L_{0}(m_{0}, \Phi(m_{0}),T_{m_{0}}\Phi) \, \vol_{g},
\end{equation*}
where
\begin{equation*}
  L_{0}(m_{0}, \Phi(m_{0}),T_{m_{0}}\Phi) := J^{g}_{\Phi}(m_{0}) \, \widetilde{L}(g(\Phi(m_{0})),\widetilde{\Psi}(\Phi(m_{0})),T_{\Phi(m_{0})}\widetilde{\Psi}),
\end{equation*}
and the Jacobian $J^{g}_{\Phi}$ is defined intrinsically by the relation $\Phi^{*}\vol_{g} = J^{g}_{\Phi} \, \vol_{g}$ ($J^{g}_{\Phi}$ depends on $g\circ \Phi$ and $T\Phi$).

Now, since $\widetilde{\Psi} \circ \Phi = \pi$, where $\pi \colon \mM \to \espace$ is the projection map $m \mapsto \xx$, we get more precisely
\begin{equation}\label{eq:L0expl}
  L_{0}(m_{0}, \Phi(m_{0}),T_{m_{0}}\Phi) = J^{g}_{\Phi}\big(g(\Phi(m_{0})), T_{m_{0}}\Phi\big) \, \widetilde{L}(g(\Phi(m_{0})),\pi(m_{0}),T_{m_{0}}\pi.(T_{m_{0}}\Phi)^{-1}),
\end{equation}
and one has then
\begin{equation*}
  \pd{L_{0}}{T\Phi} = -J^{g}_{\Phi}(T\Phi)^{-\star}\,(T\pi)^{\star}\left(\pd{\widetilde{L}}{T\widetilde{\Psi}}\circ \Phi\right)(T\Phi)^{-\star} + L_{0}\,(T\Phi)^{-\star}.
\end{equation*}
We have thus shown that the Noether stress-energy tensor for the deformation $\Phi$,
\begin{equation}\label{eq:NoetherPhi}
  \bTN{\Phi} := (T\Phi)^{\star} \pd{L_{0}}{T\Phi} - L_{0}\,\Idq^{\star},
\end{equation}
has for expression
\begin{equation*}
  \bTN{\Phi} = -J^{g}_{\Phi}\, (T\pi)^{\star}\left(\pd{\widetilde{L}}{T\widetilde{\Psi}}\circ \Phi\right)(T\Phi)^{-\star},
\end{equation*}
and using the identity $T\widetilde{\Psi} \circ \Phi = T\pi.(T\Phi)^{-1}$, we get
\begin{equation*}
  \bTN{\Phi} = -J^{g}_{\Phi} \, (T\Phi)^{\star}\left( (T\widetilde{\Psi})^{\star}\circ ^{}\Phi\right)\left(\pd{\widetilde{L}}{T\widetilde{\Psi}}\circ \Phi\right)(T\Phi)^{-\star}.
\end{equation*}
Therefore, we get finally
\begin{equation}\label{eq:TPhi}
  \bTN{\Phi} = -J^{g}_{\Phi} (T\Phi)^{\star} \left(\bTN{\widetilde{\Psi}} \circ \Phi\right) (T\Phi)^{-\star} - L_{0}\, \Idq^{\star},
\end{equation}
which we shall interpret as the (four-dimensional) \emph{relativistic Eshelby tensor}.

In this context, we shall introduce, as in Classical Continuum Mechanics, the (four-dimensional) \emph{relativistic first Piola--Kirchhoff tensor}
\begin{equation}\label{eq:relativistic-first-PK}
  \bPK := -\pd{L_{0}}{T\Phi} = J^{g}_{\Phi} \left( \bTN{\Psi} \circ \Phi\right) (T\Phi)^{-\star},
\end{equation}
and the (four-dimensional) \emph{relativistic Mandel tensor}
\begin{equation}\label{eq:relativistic-Mandel}
  \bMan := (T\Phi)^{\star}\, \bPK = J^{g}_{\Phi}\, (T\Phi)^{\star} \left( \bTN{\Psi} \circ \Phi\right) (T\Phi)^{-\star} = J_{\Phi} \,\Phi^{*}\,\bTN{\Psi} .
\end{equation}
We get then the (four-dimensional) \emph{relativistic Eshelby  tensor}
\begin{equation}\label{eq:relativistic-Eshelby}
  \bEsh := \bTN{\Phi} = (T\Phi)^{\star} \pd{L_{0}}{T\Phi} - L_{0}\,\Idq^{\star} =  -(\bMan + L_{0}\, \Idq^{\star}).
\end{equation}

\subsection{Formulation of configurational forces balance in Minkowski spacetime}

From now on, we assume that the static spacetime under consideration is the Minkowski spacetime, the one of Special Relativity. We get thus that the Universe is $\mM = \RR^{4}$ endowed with the flat metric
\begin{equation*}
  g = \eta = -c^{2} \rd t^{2}+\bq.
\end{equation*}
The quotient space $\espace$ is just $\RR^{3}$ and $\bar{g} = \bq$ is the Euclidean metric. The trivialization~\eqref{eq:universe-trivialization} is then an isometry. We will work in the global coordinate canonical system of $\RR^{4}$ for which the components of the Minkowski metric $\eta$ are constant,
\begin{equation*}
  \eta =
  \begin{bmatrix}
    -c^{2} & 0 & 0 & 0 \\
    0      & 1 & 0 & 0 \\
    0      & 0 & 1 & 0 \\
    0      & 0 & 0 & 1
  \end{bmatrix}.
\end{equation*}

In this context, we have the \emph{four-dimensional Piola identity}
\begin{equation*}
  \dive^{\eta} \bPK = J_{\Phi}\left(\dive^{\eta} \bTN{\widetilde{\Psi}}\right)\circ\Phi,
\end{equation*}
where $J_{\Phi}:=J^{\eta}_{\Phi}$,
and the  \emph{four-dimensional Mandel identity}
\begin{equation*}
  \dive^{\eta} \bMan = \dive^{\eta} \left((T\Phi)^{\star}\, \bPK\right) = (T\Phi)^{\star} \dive^{\eta} \bPK + \bPK:\nabla T\Phi
\end{equation*}
where
\begin{equation*}
  (\dive^{\eta}\bPK)_{\mu} =\tensor{P}{_{\mu}^{\beta}_{,\beta}}, \quad \text{and} \quad (\dive^{\eta}\bMan)_{\alpha} = \tensor{M}{_{\alpha}^{\beta}_{,\beta}}.
\end{equation*}
The proof of these identities will be omitted, since they are similar to the ones in 3D Classical Continuum Mechanics (see \cite{MT1992,Mau1993}).

Note furthermore that the Lagrangian density $L_{0}$ depends explicitly on $m_{0}$, only through $\xx_{0}$ but not through time $t$, because $\pi(m_{0}) = \xx_{0}$. Moreover, since $g(\Phi(m_{0})) = \eta$ is constant in our setting, $L_{0}$ does not depend on the punctual value of the 4D deformation~$\Phi$, but only on $T\Phi$. Therefore, we have
\begin{align*}
  \dive^{\eta} \bMan & = \dive^{\eta}(\bMan + L_{0}\, \Idq^{\star}) - \rd L_{0}
  \\
                     & = - \dive^{\eta} \bTN{\Phi} - \pd{L_{0}}{(t, \xx_{0})} - \pd{L_{0}}{T\Phi} : \nabla T\Phi
  \\
                     & = - \dive^{\eta} \bTN{\Phi} - \pd{L_{0}}{(t, \xx_{0})} + \bPK : \nabla T\Phi,
\end{align*}
so that
\begin{equation*}
  \dive^{\eta} \bTN{\Phi} + \pd{L_{0}}{(t, \xx_{0})} =  \bPK : \nabla T\Phi - \dive^{\eta} \bMan = - (T\Phi)^{\star} \dive^{\eta} \bPK = -J_{\Phi}\, (T\Phi)^{\star}\left(\dive^{\eta} \bTN{\widetilde{\Psi}}\right)\circ\Phi.
\end{equation*}

In conclusion, we get the equivalence
\begin{equation*}
  \dive^{\eta}\bTN{\tilde \Psi}=0 \iff  \dive^{\eta} \bTN{\Phi} + \pd{L_{0}}{(t, \xx_{0})} = 0,
\end{equation*}
where, since $\bEsh:=\bTN{\Phi}$, the equality
\begin{equation}\label{eq:relConfig}
  \dive^{\eta} \bEsh+ \pd{L_{0}}{(t, \xx_{0})} = 0
\end{equation}
is interpreted as the relativistic balance for configurational forces, with here
\begin{equation*}
  \pd{L_{0}}{(t, \xx_{0})} =
  \begin{bmatrix}
    0 & \displaystyle \pd{L_{0}}{\xx_{0}}
  \end{bmatrix},
\end{equation*}
since $L_{0}$ does not depend on time.

By \eqref{eq:relativistic-Eshelby} and \eqref{eq:TPhiblock}, the relativistic Eshelby tensor can be expressed as in Classical Continuum Mechanics (see \autoref{subsec:classical-configurational-forces}) as
\begin{equation}\label{eq:Esh4D}
  \bEsh = \bTN{\Phi} =
  \begin{bmatrix}
    \displaystyle \frac{\partial L_{0}}{\partial (\partial_{t}\phi)} . \partial_{t}\phi - L_{0}
     & \displaystyle \partial_{t}\phi  .\displaystyle\pd{L_{0}}{\bF_{\phi}}
    \\
    \bF_{\phi}^{\star} \cdot \displaystyle \pd{L_{0}}{(\partial_{t}\phi)}
     & \displaystyle \bF_{\phi}^{\star} \cdot \pd{L_{0}}{\bF_{\phi}} - L_{0}\,\Idt^{\star}
  \end{bmatrix}
\end{equation}
with the same formal expressions for the balance equations as in \autoref{subsec:Euler-law} and remark~\ref{rem:Esh4Dclassic}.

Finally, if the Lagrangian $\mL$ is moreover general covariant, then, by remark~\ref{rem:deltaPsi}, the relativistic configurational forces balance~\eqref{eq:relConfig} is equivalent to the conservation law $\dive^{\eta} \bTH=0$ for Hilbert's stress energy tensor. This implies the balance of linear momentum in Classical Continuum Mechanics \eqref{eq:MMC-Euler--Lagrange-bis} as the classical limit $c\to \infty$ \cite{Sou1964,KD2023}.

\section{The relativistic Eshelby tensor in Special Relativity and its classical limit}
\label{sec:classical-limit}

In this final section, we detail the relativistic formulation of the configurational forces balance in the particular case of the static Minkowski spacetime (example~\ref{ex:Minkowski}, $g=\eta$, $\mg=\bq$), and show that its classical limit corresponds to the well-known expressions of \autoref{subsec:classical-configurational-forces}.

In that case, the 3D quotient spatial space is the affine Euclidean space $(\espace, \bq)$ and the trivialization map~\eqref{eq:universe-trivialization} is given by
\begin{equation*}
  \btau \colon \mM \to \RR \times \espace, \qquad m \mapsto (t, \xx),
\end{equation*}
with $t = \hat{t}(m)$ and $\xx = \pi(m)$, where $\pi\colon \mM \to \espace$ is the projection map. In this trivialization, the relativistic deformation~\eqref{eq:defPhi} is expressed as
\begin{equation*}
  \Phi\colon \RR\times \Omega_{0} \to \RR \times \espace, \qquad (t, \xx_{0}) \mapsto (t, \xx = \phi(t,\xx_{0})).
\end{equation*}

The relativistic mass density \eqref{eq:relativistic-mass-density} recasts as
\begin{equation}\label{eq:rhoMink}
  \rho:= \rho_{r} \gamma =\frac{\rho_{r} }{\sqrt{1-\frac{\norm{\uu}_{\bq}^{2}}{c^{2}}}},
\end{equation}
where  $\uu$ is the (spatial) Eulerian velocity (such as $\uu \circ \Phi=\partial_{t} \phi$), and $\norm{\uu}_{\bq}$ is its Euclidean norm.
Letting $\pp_{0}\colon \body \to \espace$ be the reference embedding, recall that the equality
$(\pp_{0})_{*}\mu = \rho_{0} \vol_{\bq}$
defines the reference mass density $\rho_{0}=\rho_{0}(\xx_{0})$.
By lemma \ref{lem:mass-conservation}, setting
\begin{equation*}
  J_{\Phi}:=J_{\Phi}^{\eta}=\det T\Phi,
\end{equation*}
mass conservation on $\Omega_{0}\subset \espace$ takes the form $\left(\rho \circ \Phi\right) J_{\Phi} = \rho_{0} \circ \pi$.

\subsection{The relativistic Eshelby tensor of a general covariant matter Lagrangian}
\label{subsec:covariant-Lagrangian}

We shall assume now that the matter Lagrangian density $\widetilde L(\eta, \widetilde{\Psi}, T\widetilde{\Psi})$ is general covariant. Then, according to Souriau's theorem~\ref{thm:Souriau}, it recasts as a function $L^{gc}(\widetilde{\Psi}, \widetilde{\bK})$ of the effective matter field $\widetilde{\Psi}$ and the effective conformation
\begin{equation*}
  \widetilde{\bK} = (T\widetilde{\Psi})\, \eta^{-1} (T\widetilde{\Psi})^{\star},
  \qquad
  \widetilde{\bK} \circ \Phi
  =  \bF_{\phi}^{-1} \Big(\bq^{-1}-  \frac{1}{c^{2}} \VV \otimes \VV \Big)\bF_{\phi}^{-\star}.
\end{equation*}

In General Relativity, mass and energy are mixed together and the contribution of mass to the total energy if of order $c^{2}$. This leads to some difficulties when one tries to derive, from General Relativity to Classical Mechanics, balance equations involving energy. In~\cite{Sou1958,DeW1962}, the following Lagrangian density has been formulated
\begin{equation}\label{eq:LtildeRG}
  \widetilde L(\eta, \widetilde{\Psi}, T\widetilde{\Psi}) = \rho_{r} c^{2} + \rho_{r} w,
\end{equation}
where $\rho_{r}=\rho_{r}(\widetilde{\Psi}, \widetilde{\bK})$ is the rest mass density and $w =w(\widetilde{\Psi}, \widetilde{\bK})$ is the (specific) strain energy density of perfect (hyperelastic) matter. We deduce then, by \eqref{eq:L0expl}, lemma~\ref{lem:mass-conservation} and the definition \eqref{eq:rhoMink} of the relativistic mass density, that
\begin{equation*}
  L_{0}(\xx_{0}, \partial_{t} \phi, \bF_{\phi}) = \rho_{0}(\xx_{0}) \, \sqrt{1-\dfrac{1}{c^{2}}\norm{\partial_{t} \phi}_{\bq}^{2}}
  \left[c^{2}+ w \Big( \xx_{0}, \bF_{\phi}^{-1} \Big(\bq^{-1}-  \frac{1}{c^{2}} \partial_{t} \phi \otimes \partial_{t} \phi \Big)\bF_{\phi}^{-\star} \Big)\right],
\end{equation*}
where $\bV=\partial_{t} \phi$ is the relativistic Lagrangian velocity, and $\bF_{\phi}=T\phi_{t}$ is the linear tangent map of the three-dimensional relativistic deformation $\phi_{t}$.

Introducing the \emph{relativistic kinematic energy density},
\begin{equation}\label{eq:Krelat}
  K(\xx_{0}, \partial_{t} \phi) := \rho_{0}c^{2} \left(1-\sqrt{1-\dfrac{1}{c^{2}}\norm{\partial_{t} \phi}_{\bq}^{2}}\right) =  \frac{1}{2} \rho_{0} \norm{\VV}^{2}_{\bq} + O\Big( \frac{1}{c^{2}}\Big),
\end{equation}
and the \emph{relativistic strain energy density}
\begin{equation}\label{eq:Wrelat}
  \begin{aligned}
    W(\xx_{0}, \partial_{t} \phi, \bF_{\phi}) & := \rho_{0}\, w \Big( \xx_{0}, \bF_{\phi}^{-1} \Big(\bq^{-1}-  \frac{1}{c^{2}} \partial_{t} \phi \otimes \partial_{t} \phi \Big)\bF_{\phi}^{-\star} \Big) \sqrt{1-\dfrac{1}{c^{2}}\norm{\partial_{t} \phi}_{\bq}^{2}}
    \\
                                              & = \rho_{0}\, w \left( \xx_{0}, (\bF_{\phi}^{\star}\,\bq\, \bF_{\phi})^{-1}\right)+ O\Big( \frac{1}{c^{2}}\Big),
  \end{aligned}
\end{equation}
where we have assumed that $w=O(1)$ for perfect matter, the Lagrangian density $L_{0}$ recasts as
\begin{equation*}
  L_{0} = \rho_{0} c^{2} - K + W.
\end{equation*}
The problem is that this quantity diverges as $c \to \infty$. We shall thus prefer to substitute to the original Lagrangian density, the following one
\begin{equation*}
  L_{0} \mapsto \rho_{0} c^{2} - L_{0} = K - W,
\end{equation*}
which is moreover in accordance with the usual definition of a Lagrangian density in Classical Mechanics. Note that,
\begin{equation*}
  \int_{\Omega_{0}} L_{0} \vol_{\bq} =  m c^{2} - \int_{\Omega_{0}} \left(K - W\right) \vol_{\bq},
\end{equation*}
where $m$ (the total mass of the considered continuous medium) is a constant. Hence the quantity $\rho_{0} c^{2}$ does not interfere in the variational calculus. This redefinition of the Lagrangian density is possible in any static spacetime, using the formulation of mass conservation provided in \autoref{subsec:mass-conservation}.

\begin{rem}
  There is a rate effect, related to the  \emph{relativistic length contraction} in $1/c^{2}$, on the strain energy density $W$ (that $W\to 0$ when $\norm{\VV}_{\bq} \to c$), due to the Lorentz factor and the definition of the inverse of the relativistic right Cauchy--Green tensor,
  \begin{equation*}
    \bC^{-1} =\bF_{\phi}^{-1} \Big(\bq^{-1}-  \frac{1}{c^{2}} \VV \otimes \VV \Big)\bF_{\phi}^{-\star}.
  \end{equation*}
\end{rem}

\begin{rem}
  The fact that $L_{0}$ does not depend explicitly on the deformation $\phi$ itself, and thus that $\partial L_{0}/\partial \phi=0$, derives from the fact that we were able to find a coordinate system in which the components of the metric are constant functions. This is a consequence of Riemann theorem when the metric is flat and is particular to Special Relativity. It's consequence, in Classical Mechanics, when taking the limit $c \to \infty$, is the existence of so-called \emph{Galilean frames}.
\end{rem}

Using this choice of Lagrangian density $L_{0} := K - W$, we introduce the following quantities.
\begin{align*}
  E_{0}    & :=  \frac{\partial L_{0}}{\partial (\partial_{t}\phi)} \cdot \partial_{t}\phi - L_{0} =\bp \cdot \VV - K+W
           &                                                                                                                                                             & \text{(Relativistic energy density)},
  \\
  \bp      & := \pd{L_{0}}{(\partial_{t}\phi)}= \pd{K}{(\partial_{t}\phi)}- \pd{W}{(\partial_{t}\phi)}
           &                                                                                                                                                             & \text{(Relativistic linear momentum density)},
  \\
  \hat \bP & := - \pd{L_{0}}{\bF_{\phi}}=\pd{W}{\bF_{\phi}}
           &                                                                                                                                                             & \text{(Relativistic first Piola-Kirchhoff stress tensor)},
  \\
  \VV      & := \partial_{t} \phi
           &                                                                                                                                                             & \text{(Relativistic Lagrangian velocity)},
  \\
  \hat \bB & :=  \displaystyle \bF_{\phi}^{\star} \cdot \pd{L_{0}}{\bF_{\phi}} - L_{0}\,\Idt^{\star}= (W-K)\,\Idt^{\star}- \displaystyle \bF_{\phi}^{\star}  \, \hat \bP
           &                                                                                                                                                             & \text{(Relativistic Eshelby stress tensor)}.
\end{align*}
They allow us to recasts the \emph{relativistic Eshelby tensor} \eqref{eq:relativistic-Eshelby} as the time-space block expression
\begin{equation*}
  \bEsh  =  \begin{bmatrix}
    E_{0}                    & -\VV \cdot  \hat \bP
    \\
    \;  \bp \cdot \bF_{\phi} &
    \hat \bB
  \end{bmatrix}.
\end{equation*}

\begin{rem}
  The 3D (mixed) second order stress tensor $\hat \bB \colon T^{\star} \Omega_{0}~\to~T^{\star} \Omega_{0}$ is the relativistic generalization of the classical Eshelby stress tensor \eqref{eq:B3D-classical}.
\end{rem}

The expression~\eqref{eq:relConfig} of the relativistic balance for configurational forces recasts, as in remark~\ref{rem:Esh4Dclassic}, as
\begin{equation*}
  \pd{(K-W)}{(t, \xx_{0})}+ \dive^{\eta}  \bEsh=0,
\end{equation*}
where $\dive^{\eta} = \dive^{4D}$ in the Cartesian coordinate system $(t,X^{I})$ on the reference configuration system $\Omega_{0}$, uses both in Special Relativity and in Classical Mechanics.

The three-dimensional energy balance and configurational forces balance
\begin{equation*}
  \dd{E_{0}}{t} = \dive^{\bq} (\VV \cdot \hat \bP), \qquad \frac{\partial (K-W)}{\partial \xx_{0}} + \dive^{\bq}\, \hat \bB + \dd{}{t} \left(\bp \cdot \bF_{\phi}\right)=0,
\end{equation*}
have thus the same form in Special Relativity Hyperelasticity and in Classical Hyperelasticity~\eqref{eq:ConfigBalance}.

\begin{rem}
  In Special Relativity, the variational equations $(a)$ with respect to the reference configuration $\pp_{0}$, $\delta_{\pp_{0}} \mL=0$, and $(b)$ with respect to the effective matter field, $\delta_{\widetilde \Psi} \mL=0$, are equivalent, since the formal calculations of \autoref{subsec:classical-configurational-forces} still hold.
\end{rem}

\subsection{Explicit calculation of the components of $\bEsh$ and of their classical limit}

Let
\begin{equation*}
  \bC:= \widetilde \bK^{-1} \circ \Phi= \bF_{\phi}^{\star}\left(\bq^{-1}-  \frac{1}{c^{2}} \VV \otimes \VV \right)^{-1}\bF_{\phi}
\end{equation*}
be the generalized right Cauchy-Green tensor \eqref{eq:genC} and $\VV^{\flat}:=\eta \VV=\bq \VV$ be the covector Lagrangian velocity. We have moreover
\begin{align*}
  \pd{}{(\partial_{t}\phi)}  \sqrt{1-\dfrac{1}{c^{2}}\norm{\partial_{t} \phi}_{\bq}^{2}}
   & =-\frac{1}{c^{2}}\, \frac{\VV^{\flat}}{\sqrt{1-\frac{\norm{\VV}_{\bq}^{2}}{c^{2}}}},
  \\
  \pd{}{(\partial_{t}\phi)} w \Big( \xx_{0}, \bC^{-1} \Big)
   & =- \frac{2}{c^{2}}\,\VV\cdot\bF_{\phi}^{-\star}\pd{w}{\bC^{-1}} \,\bF_{\phi}^{-1},
  \\
  \pd{}{\bF_{\phi}} w \Big( \xx_{0}, \bC^{-1} \Big)
   & = -2\, \bF_{\phi}^{-\star} \pd{w}{\bC^{-1}} \,\bC^{-1},
\end{align*}
the first two converging to zero at $c\to \infty$. We deduce from these properties the following results.

\begin{itemize}
  \item The relativistic generalization of the first Piola-Kirchhoff stress tensor is
        \begin{equation*}
          \hat \bP = \pd{W}{\bF_{\phi} }
          = -2\rho_{0}  \sqrt{1- \frac{1}{c^{2}}\norm{\VV}^{2}_{\bq}}\;\, \bF_{\phi}^{-\star} \pd{w}{\bC^{-1}} \,\bC^{-1}.
        \end{equation*}
        This 3D constitutive equation is interpreted as the Relativistic Hyperelasticity law (as already noticed it is rate dependent by a relativistic effect, in $1/c^{2}$).
        At infinite speed of light, it converges to the Classical Hyperelasticity law
        \begin{equation*}
          \lim_{c\to \infty} \hat \bP = \pd{W}{\bF_{\phi} }
          = -2\rho_{0}\,  \bF_{\phi}^{-\star} \pd{w}{\bC^{-1}} \,\bC^{-1},
        \end{equation*}
        where $  \bC= \bF_{\phi}^{\star} \,\bq\,  \bF_{\phi}$ at the classical limit.

  \item The relativistic linear momentum density is
        \begin{equation*}
          \bp = \pd{(K-W)}{(\partial_{t}\phi)}
          = \left(1+\frac{W}{\rho_{0} c^{2}\sqrt{1 - \frac{\norm{\VV}^{2}_{\bq}}{c^{2}}}}\right) \frac{ \rho_{0} \VV^{\flat}}{\sqrt{1 - \frac{\norm{\VV}^{2}_{\bq}}{c^{2}}}}
          - \frac{1}{c^{2}}\, \VV \cdot \hat \bP\, \bC\,\bF_{\phi}^{-1},
        \end{equation*}
        since $W$ depends on the Lagrangian velocity $\VV = \partial_{t} \phi$ and $\bC\to \bF_{\phi}^{\star} \,\bq\,  \bF_{\phi}$, $\bp$ converges to the classical expression $\rho_{0} \VV^{\flat}$ at $c\to \infty$.

  \item The relativistic energy density is
        \begin{equation*}
          E_{0} = \bp \cdot \VV - K+W = \frac{K}{\sqrt{1 - \frac{\norm{\VV}^{2}_{\bq}}{c^{2}}}}
          + \frac{W}{1 - \frac{\norm{\VV}^{2}_{\bq}}{c^{2}}}
          - \frac{1}{c^{2}}\, \VV \cdot \hat \bP\, \bC\,\bF_{\phi}^{-1} \VV.
        \end{equation*}
        Setting $\bC = \bF_{\phi}^{\star} \,\bq\,  \bF_{\phi}$, and since then
        \begin{equation*}
          \lim_{c\to \infty} K = \frac{1}{2} \lim_{c\to \infty} (\bp\cdot \VV) = \frac{1}{2} \rho_{0} \norm{\VV}_{\bq}^{2},
          \qquad
          \lim_{c\to \infty} W = \rho_{0} w \big( \xx_{0}, \bC^{-1} \big) ,
        \end{equation*}
        the relativistic energy density has for classical limit
        \begin{equation*}
          \lim_{c\to \infty} E_{0} = \lim_{c\to \infty} K+\lim_{c\to \infty} W = \frac{1}{2} \rho_{0} \norm{\VV}_{\bq}^{2}+\rho_{0} w \big( \xx_{0}, \bC^{-1} \big).
        \end{equation*}

  \item The relativistic expression for $\hat \bB$ is formally the same as in Classical Continuum Mechanics, namely
        \begin{equation*}
          \hat \bB = (W-K)\,\Idt^{\star} - \displaystyle \bF_{\phi}^{\star} \, \hat \bP,
          \qquad
          \lim_{c\to \infty} \hat \bB = \left( \lim_{c\to \infty} W - \lim_{c\to \infty} K\right) \Idt^{\star} - \displaystyle \bF_{\phi}^{\star} \, \lim_{c\to \infty} \hat \bP,
        \end{equation*}
        with the above classical limits for $W$, $K$ and $\hat \bP$.
\end{itemize}

\begin{rem}
  Dust is the particular case $L=\rho_{r} c^{2}$ ($w=0$, $L_{0}=K$), for which
  \begin{equation*}
    E_{0} = \frac{K}{\sqrt{1- \frac{\norm{\VV}^{2}_{\bq}}{c^{2}}}} = \frac{1}{2} \rho_{0} \norm{\VV}_{\bq}^{2} + O\left(\frac{1}{c^{2}}\right),
    \qquad
    \bp = \frac{ \rho_{0} \VV^{\flat}}{\sqrt{1 - \frac{\norm{\VV}^{2}_{\bq}}{c^{2}}}} = \rho_{0}  \VV^{\flat} + O\left(\frac{1}{c^{2}}\right).
  \end{equation*}
  The relativistic first Piola Kirchhoff stress tensor vanishes ($\hat \bP = 0$), and the relativistic Eshelby stress tensor
  \begin{equation*}
    \hat \bB = -K\, \Idt^{\star}=-\rho_{0}c^{2} \left(1 - \sqrt{1 - \frac{\norm{\VV}_{\bq}^{2}}{c^{2}}}\right)\, \Idt^{\star}
    = -\frac{1}{2} \rho_{0} \norm{\VV}_{\bq}^{2}\, \Idt^{\star} + O\left(\frac{1}{c^{2}}\right)
  \end{equation*}
  represents purely kinematic energy density.
\end{rem}

\section*{Conclusion}

This article proposes a unified and rigorous theoretical framework for understanding and deriving configurational forces by placing them at the core of a geometric variational formulation inspired by General Relativity. Long regarded as conceptually subtle objects in Continuum Mechanics, configurational forces are clarified here through an intrinsic approach in which the material body is modeled as an abstract manifold, independent of its spatial configurations. Within the classical three-dimensional framework of hyperelasticity, we show that the balance of configurational forces naturally follows from the stationarity of the Lagrangian with respect to variations of the reference configuration. This approach highlights that the configurational force balance is not a new governing equation, but rather a reformulation of the balance of linear momentum combined with the constitutive relations, leading in a consistent manner to the Eshelby stress tensor defined on the material configuration. The main contribution of the paper lies in the extension of these results to the four-dimensional relativistic setting. Relying on a Generally Covariant formulation of matter inspired by Souriau and on the variational principles of General Relativity, the article establishes a deep correspondence between variations with respect to the matter field and variations with respect to the metric of the Universe. This yields an equivalence between the Hilbert and Noether energy–momentum tensors and provides a relativistic interpretation of the Eshelby tensor as a natural object associated with material symmetries. The introduction of an observer allows one to connect this relativistic formalism with Classical Mechanics, by identifying the non-relativistic limits and demonstrating the continuity between the classical three-dimensional and relativistic four-dimensional formulations. In this way, the balance of configurational forces emerges as a particular case, in the sense of a classical limit, of a more general relativistic conservation law. In conclusion, this work offers a unifying perspective on configurational forces, showing that they fully belong to the general framework of variational principles and conservation laws, both in Classical Mechanics and in Relativity. This approach clarifies their conceptual status, strengthens their mathematical foundations, and opens the way for future developments in the study of defects, interfaces, and complex materials within broader geometric and relativistic frameworks.

\appendix

\section{First order variational calculus and stress--energy tensors}
\label{sec:variational-calculus}

In variational calculus, one starts with a Lagrangian
\begin{equation*}
  \mL (\psi) = \int \omega(\psi),
\end{equation*}
where $\psi$ is a vector valued function from a $d$-dimensional manifold $M$ to some $r$-dimensional vector space $V$ (or more generally a section of a vector bundle above $M$ of rank $r$) and $\omega$ is a differential form of degree $d$ defined on $M$ and depending on $\psi$. The domain of integration is not specified. Usually, the dependence on $\psi$ is local, meaning that $\omega$ is only a function of the $k$-jet of $\psi$. In the following we restrict to the case where $k=1$. Hence, given some local coordinates $(x^{\mu})$ on $M$, defined on some chart $U\subset \RR^{d}$ and where $(X^{I})$ are linear coordinates on $V$, the problem recasts as
\begin{equation*}
  \mL_{U} (\psi) = \int_{U} L\left(\xx,\psi^{I}(\xx),\partial_{\nu}\psi^{I}(\xx)\right)\, \rd x^{1} \wedge \dotsb \wedge \rd x^{d}.
\end{equation*}
Here, the Lagrangian density $L = L(x^{\mu}, \psi^{I},{\psi^{I}}_{,\nu})$ is just the component of $\omega$ in the chart $(x^{\mu})$. The critical points of $\mL$ are, by definition, the functions $\psi$ such that
\begin{equation*}
  \left.\frac{\rd}{\rd\epsilon}\right|_{\varepsilon = 0}  \mL_{U} (\psi + \varepsilon \delta\psi) = 0,
\end{equation*}
for all charts $U$ and variations $\delta\psi$ with compact support in $U$. They are the solutions of the Euler--Lagrange equations
\begin{equation}\label{eq:EL-equations}
  \text{EL}_{\psi} := \pd{L}{\psi} - \dive\left(\pd{L}{T\psi}\right) = 0,
\end{equation}
where, in components,
\begin{equation*}
  \dive\left(\pd{L}{T\psi}\right)_{I} := D_{\mu} \left(\pd{L}{T\psi}\right)\indices{_{I}^{\mu}} = D_{\mu} \left(\pd{L}{(T\psi)^{I}_{\mu}}\right), \qquad 1 \le I \le r,
\end{equation*}
and $D_{\mu}$ is the \emph{total derivative} (see~\cite[chapter 4]{Olv1993}) defined as
\begin{equation*}
  D_{\mu} \left(\pd{L}{(T\psi)^{I}_{\nu}}\right) := \pd{}{x^{\mu}} \left(\pd{L}{(T\psi)^{I}_{\nu}}\left(\xx,\psi^{I}(\xx),\partial_{\alpha}\psi^{I}(\xx)\right)\right).
\end{equation*}

The conservation of the total energy density for a one-dimensional problem ($d=1$) extends for higher dimension problems, once one introduces the \emph{Noether stress--energy tensor}
\begin{equation*}
  \TNcomp{\psi}\indices{_{\mu}^{\nu}} := (T\psi)\indices{^{I}_{\mu}}\left(\pd{L}{T\psi}\right)\indices{_{I}^{\nu}} - L \delta\indices{_{\mu}^{\nu}}.
\end{equation*}
Indeed, we have the following result.

\begin{thm}\label{thm:Noether-SE-tensor-conservation}
  Let
  \begin{equation}\label{eq:Noether-SE}
    \bTN{\psi}:= (T\psi)^{\star} \cdot \pd{L}{T\psi} - L\,\bI^{\star}
  \end{equation}
  be the Noether stress--energy tensor and suppose that $\psi$ satisfies the Euler--Lagrange equations~\eqref{eq:EL-equations}. Then, we get
  \begin{equation}\label{eq:div-TN}
    \dive \bTN{\psi} =  - \pd{L}{\xx}.
  \end{equation}
  If moreover, $\psi$ is a submersion, then, \eqref{eq:EL-equations} and \eqref{eq:div-TN} are equivalent.
\end{thm}
\begin{proof}
  We have
  \begin{align*}
    (\dive \bTN{\psi})_{\mu} & = D_{\nu} \TNcomp{\psi}\indices{_{\mu}^{\nu}}                                                                                                                                                                              \\
                             & = \partial_{\nu}(T\psi)\indices{^{I}_{\mu}}\left(\pd{L}{T\psi}\right)\indices{_{I}^{\nu}} + (T\psi)\indices{^{I}_{\mu}} D_{\nu} \left(\pd{L}{T\psi}\right)\indices{_{I}^{\nu}} - (D_{\nu}L)\, \delta\indices{_{\mu}^{\nu}}
    \\
                             & = \frac{\partial^{2}\psi^{I}}{\partial x^{\nu} \partial x^{\mu}}\left(\pd{L}{T\psi}\right)\indices{_{I}^{\nu}} + (T\psi)\indices{^{I}_{\mu}}  \dive\left(\pd{L}{T\psi}\right)_{I} - D_{\mu}L,
  \end{align*}
  with
  \begin{equation*}
    D_{\mu}L = \pd{L}{x^{\mu}} + \pd{L}{\psi^{I}} \pd{\psi^{I}}{x^{\mu}} + \left(\pd{L}{T\psi}\right)\indices{_{I}^{\nu}} \frac{\partial^{2}\psi^{I}}{\partial x^{\mu} \partial x^{\nu}}.
  \end{equation*}
  We get thus
  \begin{align*}
    (\dive \bTN{\psi})_{\mu} & = (T\psi)\indices{^{I}_{\mu}}  \dive\left(\pd{L}{T\psi}\right)_{I} - \pd{L}{x^{\mu}} - \pd{L}{\psi^{I}} \pd{\psi^{I}}{x^{\mu}}
    \\
                             & = (T\psi)\indices{^{I}_{\mu}} \left\{ \dive\left(\pd{L}{T\psi}\right)_{I} - \pd{L}{\psi^{I}} \right\} - \pd{L}{x^{\mu}},
  \end{align*}
  which proves that \eqref{eq:div-TN} is satisfied if \eqref{eq:EL-equations} is satisfied. Suppose now that $\psi$ is a submersion and that \eqref{eq:div-TN} is satisfied, then we have
  \begin{equation*}
    (T\psi)\indices{^{I}_{\mu}} \left\{ \dive\left(\pd{L}{T\psi}\right)_{I} - \pd{L}{\psi^{I}} \right\} = 0,
  \end{equation*}
  which recasts as
  \begin{equation*}
    (T\psi)^{\star} \cdot \left\{\dive\left(\pd{L}{T\psi}\right) - \pd{L}{\psi}\right\} = 0.
  \end{equation*}
  But, if $T\psi$ is surjective, then $(T\psi)^{\star}$ is injective and we get thus that
  \begin{equation*}
    \dive\left(\pd{L}{T\psi}\right) - \pd{L}{\psi} = 0,
  \end{equation*}
  which achieves the proof.
\end{proof}

\begin{rem}
  Beware that with these definitions, in a change of chart $\varphi \colon U \to \tilde{U}$, $\text{EL}_{\psi}$ changes according to the rule $\text{EL}_{\psi} = (\det \varphi)\, \text{EL}_{\widetilde{\psi}}$ (see~\cite[Theorem 4.8]{Olv1993}. Similarly, $\bTN{\psi}$ is not really a tensor but a tensor-density.
\end{rem}

When the manifold $M$ is equipped with a given volume form, for instance the Riemannian volume form $\vol_{g}$ of a (pseudo) Riemannian metric $g$, we can rewrite globally
\begin{equation*}
  \omega(\psi) = L(\psi) \, \vol_{g},
\end{equation*}
where, this time, $L$ is a true function (not the component of a $d$-form, defined only locally). In that case, the Euler--Lagrange equations recasts as
\begin{equation}\label{eq:metric-EL-equations}
  \text{EL}_{\psi} := \pd{L}{\psi} - \dive^{g}\left(\pd{L}{T\psi}\right) = 0,
\end{equation}
where $\dive^{g}$ is the Riemannian divergence and the corresponding Noether stress--energy tensor is a true tensor field this time.

\section{Details of the calculations}
\label{sec:calculation-details}

\begin{lem}\label{lem:var-volume}
  Let $\pp_{0}\colon \body \to \espace$ a reference configuration. Then
  \begin{equation*}
    \delta ({\pp_{0}}^{\! \! *}(\vol_{\bq})) = {\pp_{0}}^{\! \! *}\big((\dive \uu) \vol_{\bq}\big),
  \end{equation*}
  where $\uu := \delta\pp_{0}\circ {\pp_{0}}^{-1}$.
\end{lem}

\begin{proof}
  Set $\bgamma_{0} := {\pp_{0}}^{\! \! *}\bq$. We have thus ${\pp_{0}}^{\! \! *}(\vol_{\bq}) = \vol_{\bgamma_{0}}$, and we get
  \begin{align*}
    \delta ({\pp_{0}}^{\! \! *}(\vol_{\bq})) & = \delta (\vol_{\bgamma_{0}}) = \frac{1}{2} \tr({\bgamma_{0}}^{-1}\delta \bgamma_{0}) \vol_{\bgamma_{0}},
  \end{align*}
  whereas
  \begin{equation*}
    \delta \bgamma_{0} = {\pp_{0}}^{\! \! *}(\Lie_{\uu}\bq), \quad \text{with} \quad \Lie_{\uu}\bq = \bq( \nabla \uu + (\nabla \uu)^{t}),
  \end{equation*}
  and $\uu := \delta\pp_{0}\circ {\pp_{0}}^{-1}$. We have therefore
  \begin{equation*}
    \tr({\bgamma_{0}}^{-1}\delta \bgamma_{0}) = \tr \big({\pp_{0}}^{\! \! *}(\bq^{-1}) \, {\pp_{0}}^{\! \! *}(\Lie_{\uu}\bq)\big)
    = {\pp_{0}}^{\! \! *} \tr \big(\nabla \uu + (\nabla \uu)^{t}\big) = 2\, {\pp_{0}}^{\! \! *} (\dive \uu) .
  \end{equation*}
  We get thus finally
  \begin{equation*}
    \delta ({\pp_{0}}^{\! \! *}(\vol_{\bq})) = {\pp_{0}}^{\! \! *} (\dive \uu) \,  \vol_{\bgamma_{0}} = {\pp_{0}}^{\! \! *}\big((\dive \uu) \vol_{\bq}\big).
  \end{equation*}
\end{proof}

\textbf{Proof of formula \eqref{eq:varphi}.}
We have
\begin{equation*}
  \delta_{\phi}\mL = \int_{t_{0}}^{t_{1}} \left( \int_{\Omega_{0}} \left(\pd{L_{0}}{\phi}\cdot \delta \phi + \pd{L_{0}}{\bF_{\phi}} :\rdx \delta \phi + \pd{L_{0}}{(\partial_{t}\phi)} \cdot \partial_{t} \delta \phi \right)\vol_{\bq}\right) \rd t .
\end{equation*}
But
\begin{equation*}
  \pd{L_{0}}{\bF_{\phi}} :\rdx \delta \phi = - \dive\left(\pd{L_{0}}{\bF_{\phi}}\right)\cdot \delta \phi + \dive\left(\delta \phi\cdot \pd{L_{0}}{\bF_{\phi}}\right),
\end{equation*}
or, in components,
\begin{equation*}
  {\left(\pd{L_{0}}{\bF_{\phi}}\right)_{i}}^{J} { \delta \phi^{i}}_{,J} = - {{\left(\pd{L_{0}}{\bF_{\phi}}\right)_{i}}^{J}}_{,J} \delta \phi^{i} + \left(\delta \phi^{i} {\left( \pd{L_{0}}{\bF_{\phi}}\right)_{i}}^{J} \right)_{,J}.
\end{equation*}
We also have
\begin{equation*}
  \pd{L_{0}}{(\partial_{t}\phi)} \cdot\partial_{t} \delta \phi = - \dd{}{t}\left(\pd{L_{0}}{(\partial_{t}\phi)}\right) \cdot\delta \phi + \dd{}{t}\left(\pd{L_{0}}{(\partial_{t}\phi)} \cdot\delta \phi\right),
\end{equation*}
where the total derivative means the time derivative along the path $\tilde\phi=\phi(t)$. Therefore, using Stokes theorem for the integration by part of the spatial term and Fubini theorem (commutation of the time integral with the spatial integral) for the time part, we get~\eqref{eq:varphi}.

\medskip

\textbf{Proof of formula \eqref{eq:p0}.}
We have
\begin{multline*}
  \delta_{\pp_{0}}\mL = \int_{t_{0}}^{t_{1}} \left( \int_{\body} \left[ \left(\pd{L_{0}}{\xx_{0}} + \pd{L_{0}}{\phi} \cdot \bF_{\phi}  + \pd{L_{0}}{\bF_{\phi}} :\rdx\bF_{\phi} + \pd{L_{0}}{(\partial_{t}\phi)} \cdot (\partial_{t}\bF_{\phi}) \right)\circ {\pp_{0}} \right] \cdot \delta \pp_{0} \,({\pp_{0}}^{\! \! *}\vol_{\bq})\right) \rd t
  \\
  + \int_{t_{0}}^{t_{1}} \left(\int_{\body} L_{0}\circ {\pp_{0}} \, \delta({\pp_{0}}^{\! \! *}\vol_{\bq}) \right) \rd t .
\end{multline*}
But, using lemma~\ref{lem:var-volume}, we get
\begin{equation*}
  \delta ({\pp_{0}}^{\! \! *}(\vol_{\bq})) = {\pp_{0}}^{\! \! *}\big(\dive (\delta\pp_{0}\circ {\pp_{0}}^{-1}) \, \vol_{\bq}\big),
\end{equation*}
and hence
\begin{align*}
  \delta_{\pp_{0}}\mL & = \int_{t_{0}}^{t_{1}} \left( \int_{\body} \left[\left(\pd{L_{0}}{\xx_{0}} + \pd{L_{0}}{\phi} \cdot\bF_{\phi} + \pd{L_{0}}{\bF_{\phi}} :\rdx\bF_{\phi} + \pd{L_{0}}{(\partial_{t}\phi)} \cdot(\partial_{t}\bF_{\phi}) \right)\circ {\pp_{0}} \right] \cdot \delta \pp_{0} \,({\pp_{0}}^{\! \! *}\vol_{\bq}) \right) \rd t
  \\
                      & \quad + \int_{t_{0}}^{t_{1}} \left( \int_{\body} {\pp_{0}}^{\! \! *}\big( L_{0} \,\dive (\delta\pp_{0}\circ {\pp_{0}}^{-1}) \, \vol_{\bq} \big) \right) \rd t
  \\
                      & = \int_{t_{0}}^{t_{1}} \left( \int_{\Omega_{0}} \left(\pd{L_{0}}{\xx_{0}} + \pd{L_{0}}{\phi} \cdot \bF_{\phi} + \pd{L_{0}}{\bF_{\phi}} :\rdx\bF_{\phi} + \pd{L_{0}}{(\partial_{t}\phi)} \cdot(\partial_{t}\bF_{\phi})  \right) \cdot \delta \pp_{0} \circ {\pp_{0}}^{-1} \, \vol_{\bq}\right) \rd t
  \\
                      & \quad + \int_{t_{0}}^{t_{1}} \left( \int_{\Omega_{0}} L_{0} \, \dive (\delta\pp_{0}\circ {\pp_{0}}^{-1}) \, \vol_{\bq} \right) \rd t .
\end{align*}
Now, using the fact that
\begin{align*}
  L_{0} \, \dive (\delta\pp_{0}\circ {\pp_{0}}^{-1}) & = - \rd L_{0} \cdot (\delta\pp_{0}\circ {\pp_{0}}^{-1}) + \dive \left(L_{0} \,\delta\pp_{0}\circ {\pp_{0}}^{-1} \right)
  \\
                                                     & = - \dive \left(L_{0}\, \Idt^{\star} \right) (\delta\pp_{0}\circ {\pp_{0}}^{-1}) + \dive \left(L_{0} \,\delta\pp_{0}\circ {\pp_{0}}^{-1} \right),
\end{align*}
and thanks to Stokes theorem, we get finally~\eqref{eq:p0}.


\end{document}